\documentclass{article}

\usepackage{PRIMEarxiv}
\usepackage[utf8]{inputenc} 
\usepackage[T1]{fontenc}    
\usepackage{hyperref}       
\usepackage{url}            
\usepackage{booktabs}       
\usepackage{amsfonts}       
\usepackage{nicefrac}       
\usepackage{microtype}      
\usepackage{lipsum}
\usepackage{fancyhdr}       
\usepackage{graphicx}  
\usepackage{subcaption}
\usepackage{amsthm}
\usepackage[table]{xcolor}
\usepackage{colortbl}
\usepackage{longtable}
\usepackage{arydshln}
\usepackage{authblk}
\usepackage{xcolor} 

\graphicspath{{media/}}     
\usepackage{amsmath}
\usepackage{amssymb}
\usepackage{wasysym}

\newtheorem{axiom}{Axiom}
\newtheorem{definition}{Definition}
\newtheorem{theorem}{Theorem}
\newtheorem{Property}{Property}
\newtheorem{Proposition}{\bf Proposition}
\newtheorem{Corollary}{Corollary}
\newtheorem{Conjecture}{Conjecture}

\title{System Information Decomposition
}
\author[1,2]{Aobo Lyu*}
\author[2]{Bing Yuan}
\author[3]{Ou Deng}
\author[4,2]{Mingzhe Yang}
\author[4,2]{Jiang Zhang*}
\affil[1]{Department of Electrical and Systems Engineering, Washington University in St. Louis, St. Louis, Missouri, United States of America, 63130}
\affil[2]{Swarma Research, Beijing, China, 102308}
\affil[3]{Graduate School of Human Sciences, Waseda University, Tokorozawa city, Saitama, Japan, 359-1192}
\affil[4]{School of Systems Science, Beijing Normal University, Beijing, China, 100875}

\begin{document}
\maketitle

\begin{abstract}

To characterize the complex higher-order interactions among variables within a system, this study introduces a novel framework, termed System Information Decomposition (SID), aimed at decomposing the information entropy of variables into information atoms based on their interrelations. Diverging from the established Partial Information Decomposition (PID) framework, which predominantly concentrates on the directional interactions stemming from an array of source variables to a single target variable, SID adopts a holistic approach, scrutinizing the interactions across all variables within the system. Specifically, we proved all the information atoms are symmetric, which means the disentanglement of unique, redundant, and synergistic information from any specific target variable. Hence, our proposed SID framework can capture the symmetric pairwise and higher-order relationships among variables. This advance positions SID as a promising framework with the potential to foster a deeper understanding of higher-order relationships within complex systems across disciplines.

\end{abstract}


\section{Introduction}
\label{sec:Introduction}

Systems Science is a multidisciplinary field investigating the relationships and interactions among internal variables within a system, with applications spanning neuroscience, biology, social sciences, engineering, and finance \cite{castellani2009sociology, bertuglia2005nonlinearity}. Complex systems are defined by many interconnected variables that engage in intricate interactions, the understanding of which is critical for predicting emergent properties, devising novel treatments, and optimizing system performance. 

In the field of information theory, mutual information is a widely employed method for quantifying interactions between two variables by encapsulating shared information or the reduction in uncertainty facilitated by each variable \cite{shannon2001mathematical}. However, mutual information is restricted to describing pairwise interactions, which often proves inadequate for analyzing complex systems that necessitate multivariate interaction assessments. 

As a solution, Williams and Beer introduced the Partial Information Decomposition (PID) method, which characterizes information interactions between a target variable and multiple source variables by decomposing the mutual information shared among them \cite{williams2010nonnegative}. In the past ten years, PID and related theories, such as Information Flow Modes \cite{james2018modes} and integrated information theory \cite{mediano2019beyond}, have been applied in many fields, such as quantitative identification of Causal Emergence \cite{rosas2020reconciling}, dynamical process analysis \cite{james2016information} and information disclosure \cite{rosas2020operational, rassouli2019data}. However, PID-related techniques only decompose the partial information of a single target variable at a time. This leads to the fact that selecting or constructing a suitable and plausible target variable can be challenging or even unfeasible when addressing complex systems problems, and also raising questions as to why certain variables are prioritized as targets over others. Moreover, this variable-specific perspective results in a unidirectional relationship between the specified target variable and source variable, which makes information atoms bound to a specific target variable and insufficient for a comprehensive description of the relationships among variables. This further limits our exploration of system functions and properties, as many of them originate from the relationship between system variables rather than specific variables or its asymmetric properties.

To overcome these limitations, we need a system analysis method based on a system perspective, analogous to the synchronization model \cite{acebron2005kuramoto} or the Ising model \cite{cipra1987introduction}, rather than a variable perspective like PID. Furthermore, this method should capture the nature and characteristics of the system without specifying or introducing any special variable, and also take into account all the interactive relationships among all variables in the system, including pairwise and higher-order relationships. Therefore, we propose System Information Decomposition (SID), an innovative method based on PID that treats all system variables equally (target-free) and effectively captures their intricate interactions. This novel approach enhances our capacity to scrutinize and understand the complexities of multivariate systems.

Specifically, we firstly expand the PID's conceptual framework to a system horizon by taking all variables in the system as target variable separately. Then,  we prove the symmetry properties of information decomposition based on a set theory perspective of information theory. That means the value of information atoms, the non-overlapping units obtained by decomposing variables' information entropy according to their relationship, will not be affected by the the choice of target variable. Therefore, we put forward a general SID framework, wherein redundant, synergistic, and unique information atoms become a multivariate system's property, reflecting the complex (pairwise and higher-order) relationships among variables. Furthermore, we explore the connections between existing information entropy indicators and the information atoms within the SID framework while proposing the necessary properties for information atom quantification and several variable calculation approaches. Through a detailed case analysis, we provide an intuitive demonstration that SID can unveil higher-order relationships within the system that cannot be captured by existing probability or information measures. Finally, we discuss the potential application scenarios and implications of SID from the philosophical perspective of system decomposition as well as from areas such as Higher-order Networks and theory of Causality.

Our contributions to Information and System Science are twofold. Firstly, the SID framework broadens the application of information decomposition methods in complex systems by introducing a methodology to decompose all variables' entropy within a system. This achievement also unifies information entropy and information decomposition onto one Venn diagram, where three variables can be well represented on a two-dimensional plane. Secondly, this framework reveals previously unexplored higher-order relationship that cannot be represented by existing probability or information measures, providing a potential data-driven quantitative framework for Higher-order Networks related research. 

The remainder of this paper is organized as follows. Section \ref{sec:Literature Review} reviews the development of information theory, PID and related research. Section \ref{sec:System Information Decomposition} extends the PID method to multivariate system scenarios, defines SID, shows the connections between existing information entropy indicators and the information atom. Section \ref{sec:Case Studies} presents the characteristics of the SID framework through a case analysis. Then, Section \ref{sec:Calculation of SID} gives the properties for information atom calculation and their possible calculation approaches. The significance and potential applications of SID are discussed in Section \ref{sec:Discussion}.

\section{Information Decomposition}
\label{sec:Literature Review}

\subsection{Information Theory Framework}
\label{sec:Information Theory Framework}

Shannon's classical information theory has provided a robust foundation for understanding information entropy \cite{shannon2001mathematical}. Mutual information and conditional entropy further decompose information and joint entropy according to the pairwise relationship between variables, which can be intuitively shown in Venn diagrams \ref{fig:Info}. In this paper, we explore the potential of Venn diagrams to provide valuable insights into the decomposition of multivariate systems and extend the entropy decomposition approach of classical information theory.

\begin{figure}[htbp]
    \centering
    \fbox{\includegraphics[width=.5\linewidth]{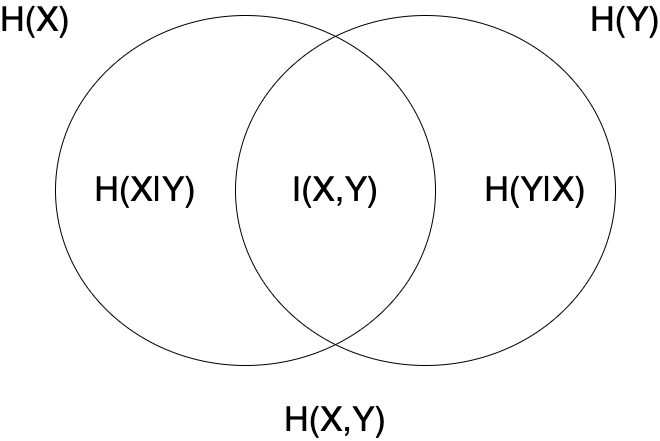}}
    \caption{Information Theory Venn Diagram.}
    \label{fig:Info}
\end{figure}

\subsection{Partial Information Decomposition Framework}
\label{sec:Partial Information Decomposition Framework}

In classical information theory, the joint mutual information may occasionally be larger or smaller than the sum of the mutual information between individual variables. Consequently, traditional redundant information calculations may yield negative values \cite{schneidman2003synergy, barrett2015exploration}, contradicting our intuitive understanding. To address this phenomenon, Williams and Beer proposed the Partial Information Decomposition (PID) framework \cite{williams2010nonnegative}.

The PID framework facilitates the decomposition of joint mutual information between multiple source variables and a target variable. Specifically, for a random target variable $Y$ and a random source variables $X={X_{1}, X_{2}, \cdots, X_{n}}$, the PID framework allows for the decomposition of the information that $X$ provides about $Y$ into information atoms, such as redundant, synergistic and unique information. These atoms represent the partial information contributed by various subsets of $X$, individually or jointly, providing a more nuanced understanding of the relationships between the target and source variables.

Considering the simplest case of a system with three variables, one can employ a Venn diagram to elucidate their interactions \cite{williams2010nonnegative}. The unique information $Un(Y:X_1)$ from $X_1$ signifies the information that $X_1$ provides to $Y$, which is not provided by $X_2$ and vice versa. In other words, unique information refers to the contribution made by a specific source variable to the target variable that is exclusive to that variable and not shared by other source variables. Redundant information $Red(Y:X_1,X_2)$ represents the common or overlapping information that $X_1$ and $X_2$ provide to $Y$. Synergistic information $Syn(Y:X_1,X_2)$ captures the combined contribution of $X_1$ and $X_2$ to $Y$, which cannot be obtained from either variable individually. Formaly, it can be write in the form: $I(Y: X_{1}, X_{2}) = Syn(Y: X_{1}, X_{2}) + Red (Y: X_{1}, X_{2})+ Un(Y: X_{1}) + Un(Y: X_{2})$.
\begin{figure}[htbp]
    \centering
    \fbox{\includegraphics[width=.5\linewidth]{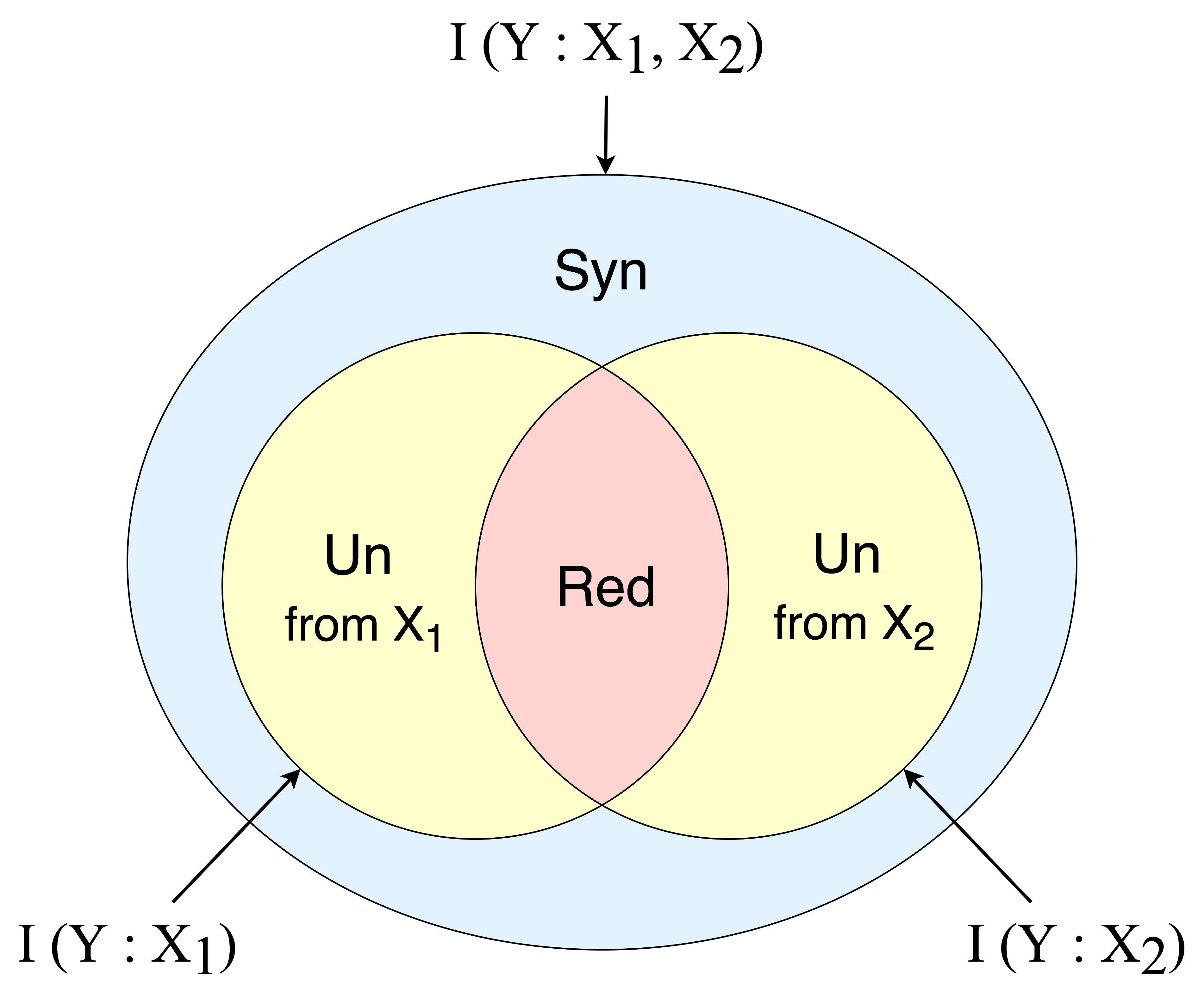}}
    \caption{Venn Diagram of PID.}
    \label{fig:PID_Venn}
\end{figure}

For an arbitrary multivariate system, we can select any variable as the target variable $Y$ and the remaining variables as the source variables ${X_{1},\cdots ,X_{n}}$. The redundant information $Red(Y:{X_{1},\cdots ,X_{n}})$ denotes the common or overlapping information provided by the source variables \cite{williams2010nonnegative}, which is contained in each source \cite{kolchinsky2022novel}.
Besides, Redundant information has the following properties \cite{williams2010nonnegative}:

\begin{axiom} [Symmetry of source variables]
\label{Axiom:Symmetry of source variables}
\(Red(Y : {X})\) is invariant to the permutation of X.
For the source variables $X_{i}$ and $X_{j}$ from $\{X_{1}, \cdots ,X_{n}\}$,$i,j \in \{1 \cdots n\} $ , there is $Red(Y: X_{i}, \cdots  X_{j}) = Red(Y: X_{j}, \cdots  X_{i})$.
\end{axiom}

\begin{axiom} [Self-redundancy]
\label{Axiom:Self-redundancy}
When there is only one source variable, the redundant information is equivalent to the mutual information between the target variable $Y$ and the source variable $X_{i}$, i.e. $Red(Y:X_{i}) = I(Y:X_{i})$.
\end{axiom}

\begin{axiom} [Monotonicity]
\label{Axiom:Monotonicity}
The redundancy should exhibit a monotonically decreasing behavior with the inclusion of additional inputs, i.e. 
$Red(Y: X_{1}, \cdots ,X_{n}) \le  Red(Y: X_{1}, \cdots ,X_{n-1})$, where $n \in N$.
\end{axiom}

Despite numerous quantitative methods for information atoms in PID, a widely accepted method still needs to be discovered, primarily due to negative solutions. Such inconsistencies undermine the notion of information entropy as a non-negative measure of uncertainty. To circumvent reliance on a specific quantitative method, we employ classical mutual information and conditional entropy for calculating the sum of the information entropy of certain information atoms. Although this approach does not permit the precise calculation of individual information atoms \cite{williams2010nonnegative, williams2011information}, it ensures that the framework remains independent of how PID atoms are calculated. Consequently, when a particular PID calculation method computes the value of one information atom, the information entropy of the remaining information atoms is determined by the following Axiom:

\begin{axiom} [Quantitative Computation] 
\label{axiom:Quantitative Computation}
In a system with a target variable $Y$ and source variables $X_{i}$ and $X_{j}$, the following relationships hold:
\begin{align}
\label{align:RedUn}
Un(Y: X_{i}) = I (X_{i}: Y) - Red(Y: X_{i}, X_{j})
\end{align}
\begin{align}
Syn(Y: X_{i}, X_{j})= H (Y | X_{j}) - H (Y | X_{i}, X_{j}) - Un(Y: X_{i}) 
\end{align}

The above rules can also be extended to general system with a target variable $Y$ and source variables $X_{1}, \cdots, X_{n}$, where $(X_{1}, \cdots, X_{n}\setminus X_{i})$ represents the union variable of $X_{1}, \cdots, X_{n} $ without$ X_{i})$.
\begin{align}
Un(Y: X_{i}) = I (X_{i}: Y) - Red(Y: X_{i}, (X_{1}, \cdots, X_{n}\setminus X_{i}))
\end{align}
\begin{align}
Syn(Y: X_{1},\cdots, X_{n})
= H (Y | (X_{1}, \cdots, X_{n}\setminus X_{i})) - H (Y | X_{1}, \cdots, X_{n}) - Un(Y: X_{i}) 
\end{align}

\end{axiom}

Although several enlightening perspectives on PID have been proposed \cite{williams2010nonnegative, griffith2014intersection, ince2017measuring, bertschinger2013shared, harder2013bivariate, bertschinger2014quantifying}, there is still no perfect calculation methods. To make our work not rely on any specific computational method, we need to explore information decomposition and the properties of information atoms from a more conceptual perspective. Given the high similarity between information decomposition, especially the concept of redundant information, and the concept of inclusion and overlapping, set theory may allow us to explore the properties of redundant more deeply.

\subsection{A Set-theoretic Understanding of PID}
\label{sec:A set-theoretic understanding of PID}

Kolchinsky's remarkable work \cite{kolchinsky2022novel} offers an understanding based on set theory. Given that PID is inspired by an analogy between information theory and set theory \cite{williams2011information}, the redundant information can be understood as information sets that the sources provide to the target. More specifically, the definition of set intersection $\cap \{X_{i}\}$ in set theory means the largest set that is contained in all of the $X_{i}$, and these set-theoretic definitions can be mapped into information-theoretic terms by treating “sets” as random variables, “set size” as entropy, and “set inclusion” as an ordering relation $\sqsubset$, which indicates when one random variable is more informative than another. 

Considering a set of sources variables \(X_{1}, . . . ,X_{n}\) and a target $Y$, PID aims to decompose $I(Y: X_{i}, X_{j})$ and get $Red(Y: X_{1}, \cdots ,X_{n})$, the total same information provided by all sources about the target, into a set of non-negative terms. Therefore, redundant information can be viewed as the "intersection" of the information contributed by different sources, leading to the following definition:

\begin{definition}[Set Intersection of Information \cite{kolchinsky2022novel} ]
    
\label{definition:Set Intersection of Information}
For a variable-system, the redundant information from the source variables $X_{1}, \cdots ,X_{n}$ to the target variable $Y$ is the information that all source variables can provide to the target variable, the largest mutual information between the target variable and a non-unique variable $Q$ that has an ordering relation $\sqsubset$ with all source variables. That is 
\begin{align}
Red(Y:X_{1},\cdots, X_{n})= I_{\cap} (X_{1},\cdots,X_{n}\to Y) &:= \sideset{}{}\sup_Q \{I(Q:Y):Q\sqsubset X_{i},\forall i\in \{1\cdots n\} \}
\end{align}
\end{definition} 

The ordering relation $\sqsubset$ is an analogy to the relation contained $\subseteq$ in set theory, which is not specified but follows some assumptions: i) Monotonicity of mutual information, $A \sqsubset B \Rightarrow I(A:Y) \le I(B:Y) $. ii) Reflexivity: $A \sqsubset A$ for all variable $A$. iii) For all sources $X_{i}$, $O \sqsubset X_{i} \sqsubset (X_{1},\cdots, X_{n})$, where $H(O) = 0$ and $(X_{1},\cdots, X_{n})$ indicates all sources considered jointly. For example, the partial order can be Q $\sqsubset$ X if and only if $H(Q|X)=0$, or the well-known Blackwell order \cite{blackwell1953equivalent}, such that $Q$ precedes $X_i$ if $X_i$ has all of the information that $Q$ has, about some third variable $Y$.

\section{System Information Decomposition}
\label{sec:System Information Decomposition}
In this section, we develop a mathematical framework of SID. The objective of this framework is to decompose the information of all variables within a system based on their interrelationships. By addressing the limitation of PID, which focuses solely on a single target variable, we progress towards multi-variable information decomposition for systems.

\subsection{Extension of PID in a System Scenario}
\label{sec:Extension of PID in system scenario}

The PID method only decomposes joint mutual information between multiple source variables and a specific target variable, as illustrated by the outermost circle of the Venn diagram in Figure \ref{fig:PID_Venn}. We redesign the Venn diagram to extend this method and encompass a system-wide perspective, as demonstrated in Figure \ref{fig:Venn_1.0}. The system comprises two source variables,$X_{1}$ and $X_{2}$, and one target variable, $Y$, represented by the three intersecting circles. 

The area size within the figure signifies the information entropy of the variables or information atoms, and the central area denotes the joint mutual information, encompassing redundant, unique from $X_{1}$, unique from $X_{2}$, and synergistic information. This arrangement aligns with the Venn diagram framework of PID.

\begin{figure}[htbp]
\centering
\fbox{\includegraphics[width=.5\linewidth]{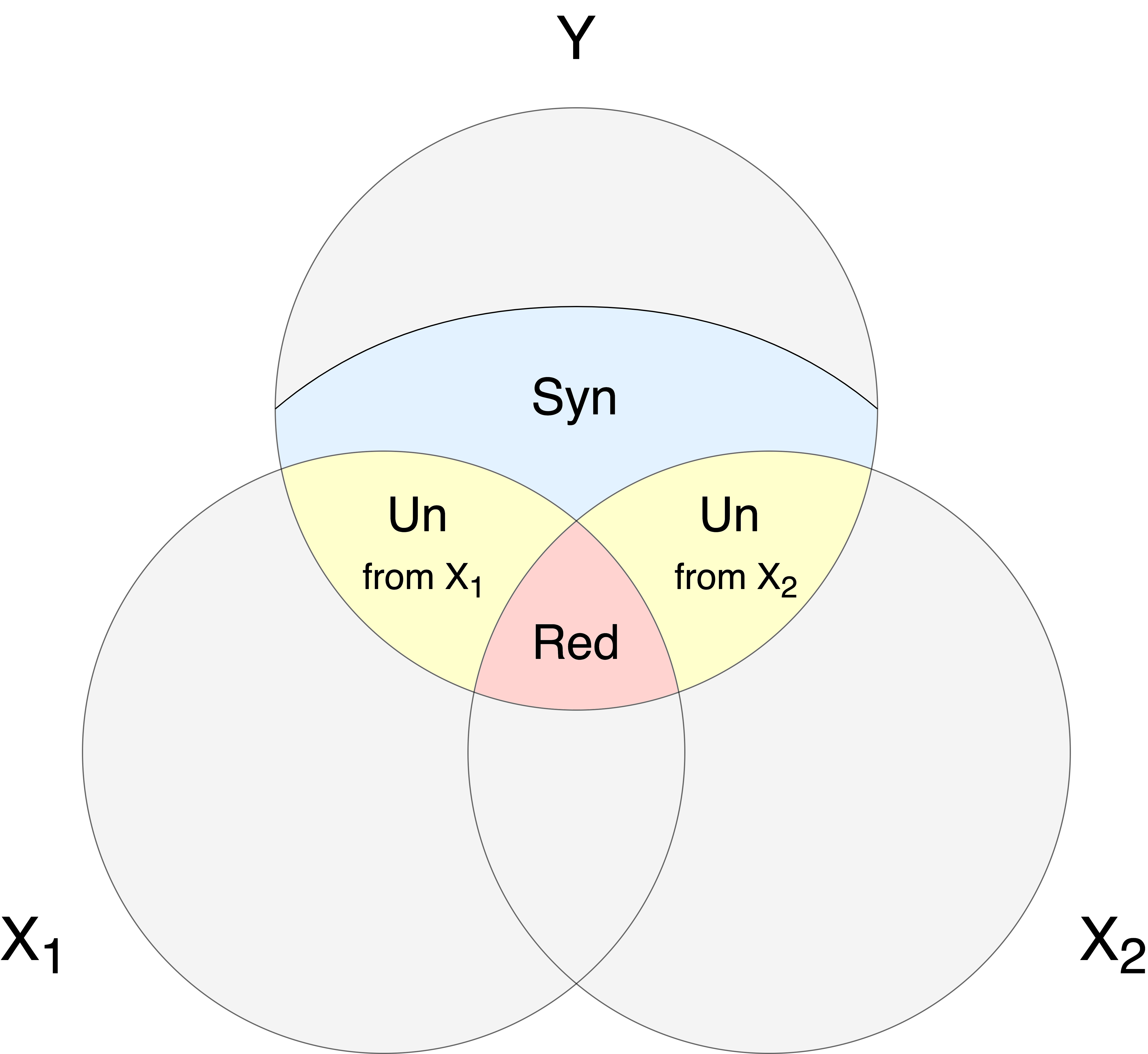}}
\caption{Venn diagram from different perspectives of PID.}
\label{fig:Venn_1.0}
\end{figure}

To enhance the comprehensiveness of the framework, it is necessary to elucidate the unexplored section of the updated Venn diagram \ref{fig:Venn_1.0}. In addition to the four sections of joint mutual information, the information entropy of the target variable $Y$ contains an unaccounted-for area. According to Shannon's formula, this area corresponds to the joint conditional entropy of the source variables to the target variable $H(Y| X_{1}, X_{2})$, which also characterizes the interrelationships between the target variable and the source variables. In the SID framework, numerous joint conditional entropy exist, including one that stands out: the joint conditional entropy originating from all variables except the target variable. To optimize the usefulness of the SID framework, we define this specific joint conditional entropy as the target variable's external information ($Ext$). The definition is grounded in the philosophical assumption that everything is interconnected. Since joint conditional entropy implies the uncertainty that cannot be eliminated by the internal variables of the system, the variables capable of providing this information must exist outside the system. To some extent, external information can emphasize the relationship between the target variable and the entire system rather than just a simple relationship with other variables. Therefore, we also consider it a kind of information atom within the SID framework.

\begin{definition}[External Information]
\label{definition:External Information}
For a system containing variables $Y$ and $\{X_{1}, \cdots, X_{n}\}$, the external information $Ext(Y)$ is defined as: 
\begin{align}
Ext(Y)=H(Y|X_{1}, X_{2}, \cdots, X_{n})
\end{align}
\end{definition}

Thus, we have been able to decompose the target variable's entropy into a finite number of non-repeated information atoms according to the relationship between it and the other variables in the system. Furthermore, we can apply this information decomposition method to each variable in the system to decompose the entire information entropy of the system, which results in a preliminary version of the SID. For the convenience of expression, we use $Un_{i-j}$, $Syn_{ij-k}$, and $Red_{ij-k}$ to represent $Un(X_j,X_i)$, $Syn(X_k:X_i,X_j)$, and $Red(X_k:X_i,X_j)$ respectively. A Venn diagram for a three-variable system is shown in Figure \ref{fig:Venn_1.5}:

\begin{figure}[htbp]
\centering
\fbox{\includegraphics[width=.5\linewidth]{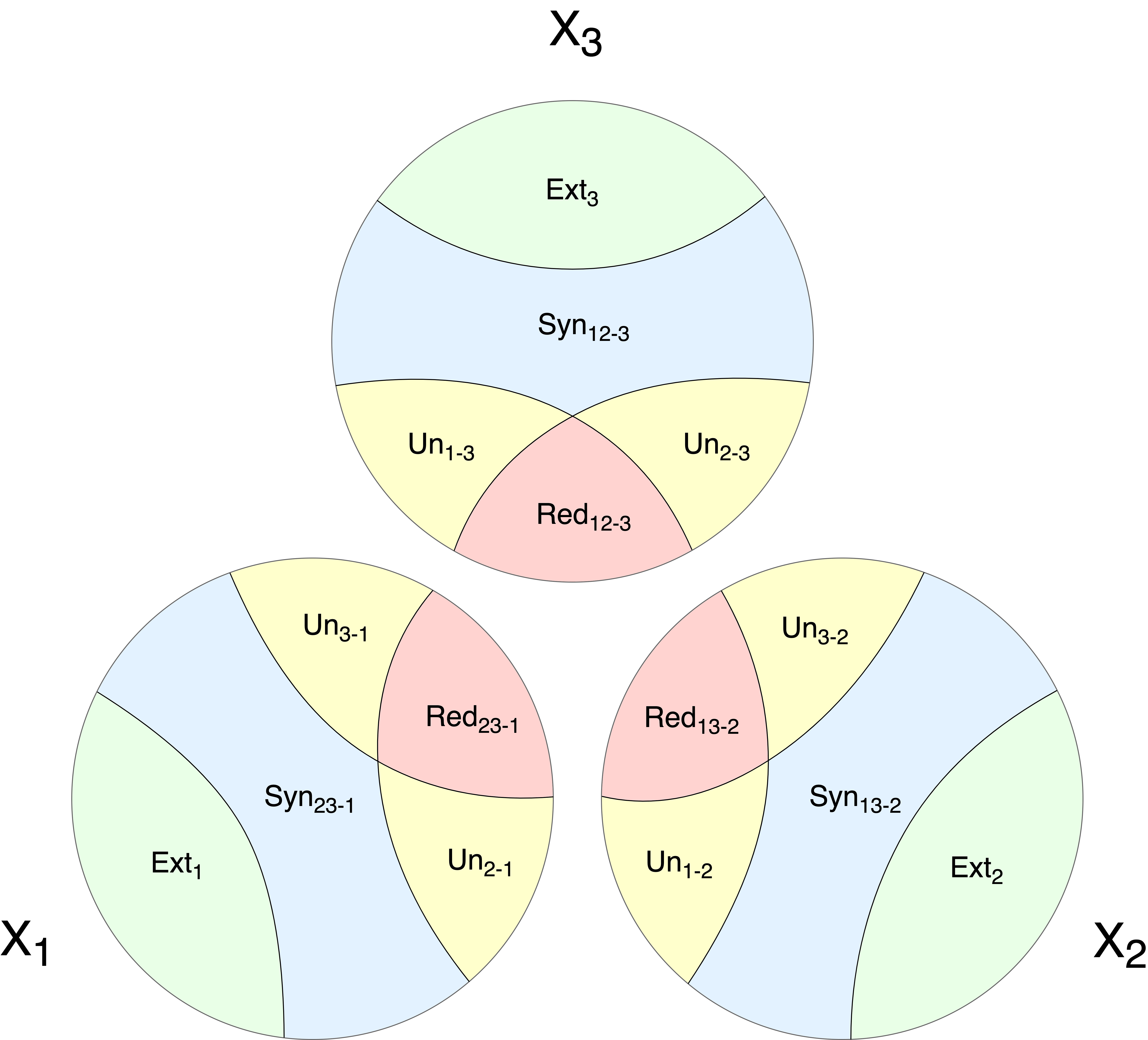}}
\caption{Venn diagram of SID's Preliminary version.}
\label{fig:Venn_1.5}
\end{figure}

\subsection{Properties of Information Atoms}
\label{sec:Properties of information atoms}

Although the preliminary version of SID can decompose all variables in a system, the decomposition of each variable is carried out separately, and the description of information atoms is directional (from source variables to the target variable). For instance, the unique information provided by $X_1$ to $X_3$ in Fig. \ref{fig:Venn_1.5} is not directly related to the unique information provided by $X_3$ to $X_1$. To make information atoms better reflect the relationship among variables and unifies the Venn diagram of Shannon's framework \ref{sec:Information Theory Framework} and PID framework \ref{sec:Partial Information Decomposition Framework}, it is necessary to further explore the properties of information atoms within the SID framework. In this subsection, we are going to prove the symmetry property of information atoms by demonstrating that unique, redundant, and synergistic information atoms remain stable when different variables are considered as target variables.

\begin{theorem}[Symmetry of Redundant Information]
\label{Theorem:Symmetry of Redundant Information}
Let $X_{1},\cdots , X_{n}$ be the variables in a system. The redundant information is equal irrespective of the chosen target variable. Formally, we write $Red(X_{i}:{X_{1},\cdots , X_{n}}\setminus X_{i}) = Red(X_{j}:{X_{1},\cdots , X_{n}}\setminus X_{j}), \forall i, j \in \{1\cdots n\}$.
\end{theorem}

To proof this Theorem, we use the Definition \ref{definition:Set Intersection of Information} with the Blackwell partial order, such that $Q$ precedes $X_i$ if $X_i$ has all of the information that $Q$ has, about the target variable $Y$, which written in the form $Q \sqsubset _Y$$X_i$. (In Appendix \ref{Proof of Theorem 1 with different partial order}, we also provide a proof based on another partial ordered.)

\begin{proof}Suppose we have a multivariate system containing a target variable $Y$ and source variables $X_{1},\cdots , X_{n}$. For the convenience of expression, we use $\mathcal{X}$ to represent all the source variables $X_{1},\cdots , X_{n}$. The proof is to show that $Red(Y: \mathcal{X},Y) = Red(Y;\mathcal{X})$ and $ Red(U: \mathcal{X}, Y) = Red(Y: \mathcal{X}, Y) $, where $U$ is the union variable of $Y$ and $\mathcal{X}$, such that $U = (\mathcal{X}, Y)$. (The entropy of the union variable $U$ can be expressed as $H(U) = H(\mathcal{X}, Y)$.) Then, we can demonstrate that redundant information is equal regardless of which variable is chosen as the target variable.

\underline{Step One, to prove $Red(Y: \mathcal{X},Y) = Red(Y:\mathcal{X}):$} 

By Definition \ref{definition:Set Intersection of Information}, 

\begin{align}
Red(Y: \mathcal{X}) =\sup_{Q_{j}} \{I(Q_{j}:Y): Q_{j}\sqsubset_Y X_{i},\forall i\in \{1\cdots n\} \} 
\end{align}
According to Blackwell order, $Q_j \sqsubset_Y Y$, since $Y$ has all of the information about $Y$. Then, we have: 
\begin{align}
\sup_{Q_{j}} \{I(Q_{j}:Y): Q_{j}\sqsubset_Y X_{i},\forall i\in \{1\cdots n\} \} = \sup_{Q_{j}} \{I(Q_{j}:Y): Q_{j}\sqsubset_Y Y, Q_{j}\sqsubset_Y X_{i}, \forall i\in \{1\cdots n\} \} 
\end{align}

Therefore, $Red(Y: \mathcal{X},Y) = Red(Y;\mathcal{X})$.

\underline{Step Two, to prove $ Red(U: \mathcal{X}, Y) = Red(Y: \mathcal{X}, Y) $:}

Building upon the conclusion that $Red(Y: \mathcal{X}, Y) = Red(Y: \mathcal{X})$, we can replace the target variable with the union variable $U = (\mathcal{X}, Y)$.

By Definition \ref{definition:Set Intersection of Information},
\begin{align}
\label{align:red}
Red(U: \mathcal{X}, Y) =\sup_{Q_{j}} \{I(Q_{j}:U):Q_{j}\sqsubset_U Y , Q_{j}\sqsubset_U X_{i} ,\forall i\in \{1\cdots n\} \} 
\end{align}

Let $Q_{j}^{*}$ satisfies or infinitely approaches the above conditions:
\begin{align}
I(Q_{j}^{*}:U) &= Red(U: \mathcal{X}, Y) - \varepsilon, \forall \varepsilon > 0 
 \nonumber \\ &= \sup_{Q_{j}} \{I(Q_{j}:U):Q_{j}\sqsubset_U Y, Q_{j}\sqsubset_U X_{i},\forall i\in \{1\cdots n\}\} - \varepsilon, \forall \varepsilon > 0, \nonumber 
\end{align}

Since $U=(\mathcal{X},Y) (H(Y|U)=0)$, then
$I(Q_{j}^{*}:U) \ge I(Q_{j}^{*}:Y) $. Considering that $Q_{j}^{*} \sqsubset_U Y$, which means $Y$ has all of the information that $Q_{j}^{*}$ has, about the target variable $U$, such that $I(Q_{j}^{*}:U) \le I(Q_{j}^{*}:Y) $, we have:
\begin{align}
I(Q_{j}^{*}:U) = I(Q_{j}^{*}:Y) \end{align}

Since $Y$ has all the information about itself, we have:
\begin{align}
Q_{j}^{*}\sqsubset_Y Y 
\end{align}

Since $U=(\mathcal{X},Y) (H(Y|U)=0)$ and $Q_{j}^{*}\sqsubset_U X_{i},\forall i\in \{1\cdots n\} \}$ ($X_i$ has all of the information that $Q_{j}^{*}$ has, about the target variable $U$), we have:
\begin{align} 
\label{align:QYX}
Q_{j}^{*}\sqsubset_Y X_{i},\forall i\in \{1\cdots n\} 
\end{align}

Therefore,  by Equation \ref{align:red}-\ref{align:QYX} and Definition \ref{definition:Set Intersection of Information}, we obtain:
\begin{align}
Red(U: \mathcal{X}, Y) &=\sup_{Q_{j}} \{I(Q_{j}:Y):Q_{j}\sqsubset_Y Y , Q_{j}\sqsubset_Y X_{i} ,\forall i\in \{1\cdots n\}\}  \nonumber\\
&= Red(Y: \mathcal{X}, Y)\nonumber
\end{align}

\underline{In Summary:}
Since we have established that $Red(Y: \mathcal{X}, Y) = Red(Y: \mathcal{X})$, and $ Red(U: \mathcal{X}, Y) = Red(Y: \mathcal{X}, Y)$, we can conclude that for all $X_{i}$ in $\{\mathcal{X}\}$, $Red(X_{i}: Y, \{\mathcal{X}\} \setminus  X_{i}) = Red(Y:\{\mathcal{X}\})$. Therefore, Theorem \ref{Theorem:Symmetry of Redundant Information} is proved, and we can use $Red(X_{1},\cdots, X_{n})$ or $Red_{1\cdots n}$ denote the redundant information within the system $\{X_{1},\cdots, X_{n}\}$.
\end{proof}

\begin{theorem} [Symmetry of Unique Information]
\label{Theorem:Symmetry of Unique Information}
Let $X_{1},\cdots, X_{n}$ be the variables in a system. In SID, the unique information of any two variables relative to each other is equal, regardless of which is chosen as the target variable. Formally, we write $Un(X_{i}:X_{j}) = Un(X_{j}:X_{i})$, $\forall i \ne j$ where $i, j \in \{1, \cdots, n\}$.
\end{theorem}

\begin{proof}
According to Axiom \ref{axiom:Quantitative Computation}, unique information is a part of the information provided by the source variable to the target variable, that is, mutual information minus redundant information. In a three-variable system $\{X_{1},X_{2},X_{3}\}$, by Equation \ref{align:RedUn} we have: 
\begin{align}
Un(X_{i}:X_{j}) = I (X_{i}; X_{j}) - Red(X_{i}:X_{j},X_{k}), \forall i \ne j \in \{1,2,3\}
\end{align}
Since $I(X_{i}:X_{j}) = I(X_{j}:X_{i})$ according to the symmetry of Shannon's formula \cite{shannon2001mathematical}, and $Red(X_{i}:X_{j},X_{k}) = Red(X_{j}:X_{i},X_{k}) = Red(X_{i},X_{j},X_{k})$ according to Theorem \ref{Theorem:Symmetry of Redundant Information}, we have:
\begin{align}
Un(X_{i}:X_{j}) &= I (X_{i}; X_{j}) - Red(X_{i}:X_{j},X_{k}) \nonumber\\
&= I (X_{j}; X_{i}) - Red(X_{j}:X_{i},X_{k})\nonumber\\
&= Un(X_{j}:X_{i})\nonumber
\end{align}
For general multivariate systems $X_1, \dots, X_n$, we can prove the symmetry of unique information between any two variables $X_i$ and $X_j$ by combining other variables $X_1, \cdots, X_n \setminus X_i,X_j $ into one variable $X_k$. Therefore, we proved the theorem, and we can represent this information atom as $Un(X_{i},X_{j})$, or $Un_{i,j}$.
\end{proof}

\begin{theorem} [Symmetry of Synergistic Information]
\label{Theorem:Symmetry of Synergistic Information}
Let $X_{1},\cdots , X_{n}$ be the variables in a system. In SID, the synergistic information of any group of variables is equal, regardless of which is chosen as the target variable. Formally, we write $Syn(X_{i}:\{X_{1},\cdots , X_{n}\}\setminus X_{i})= Syn(X_{j}:\{X_{1},\cdots , X_{n}\}\setminus X_{j}), \forall i,j\in \{1\cdots n\}$.  
\end{theorem}

\begin{proof}
According to Axiom \ref{axiom:Quantitative Computation}, Theorem \ref{Theorem:Symmetry of Unique Information}, and the chain rule of Shannon formula, for a three-variable system with ${X_{i}, X_{j}, X_{k}}$:
\begin{align}
Syn(X_{k}:X_{i},X_{j}) &= H(X_{k}|X_{j}) - H(X_{k}|X_{i},X_{j}) - Un(X_{i},X_{k}) \nonumber \\
&= (H(X_{j},X_{k}) - H(X_{j})) - (H(X_{i},X_{j},X_{k}) - H(X_{i},X_{j})) - Un(X_{i},X_{k}) \nonumber \\
&= H(X_{j},X_{k}) + H(X_{i},X_{j}) - H(X_{j}) - H(X_{i},X_{j},X_{k}) - Un(X_{i},X_{k}) \nonumber \\
&= (H(X_{i},X_{j}) - H(X_{j})) - (H(X_{i},X_{j},X_{k}) - H(X_{j},X_{k})) - Un(X_{i},X_{k}) \nonumber \\
&= H(X_{i}|X_{j}) - H(X_{i}|X_{j},X_{k}) - Un(X_{i},X_{k}) \nonumber \\
&= Syn(X_{i}:X_{j},X_{k}) \nonumber
\end{align}
Therefore, we demonstrate that $X_i$ and $X_k$ are interchangeable, and since $X_i$ and $X_j$ are interchangeable as source variables, we proved that all variables are interchangeable. For general multivariate systems $X_1,\cdots,X_n$, we can prove the symmetry of synergistic information between any two variables $X_i$ and $X_k$ by combining other variables into one variable $X_j$. Therefore, we proved Theorem \ref{Theorem:Symmetry of Synergistic Information} and we can write synergistic information in the form of $Syn(X_{1},\cdots, X_{n})$ or $Syn_{1\cdots n}$.
\end{proof}

Based on the Theorem \ref{Theorem:Symmetry of Redundant Information} \ref{Theorem:Symmetry of Unique Information} \ref{Theorem:Symmetry of Synergistic Information} (the symmetry of information atoms), the SID framework can be merged into the formal version in Figure \ref{fig:Venn_2.0}. In the formal version of SID, the concept of target variable is canceled, and all variables are equally decomposed according to their relationship with other variables. Specifically, redundant information and unique information are merged. redundant information (atoms) in any group of variables and unique information (atoms) between any two variables appear only shown one time in the Venn diagram, while synergistic information (atoms) appear in each participating variable with the same value, and each variable contains one external information (atom). So far, we can give the formal definition of SID: 

\begin{figure}[htbp]
\centering
\fbox{\includegraphics[width=.5\linewidth]{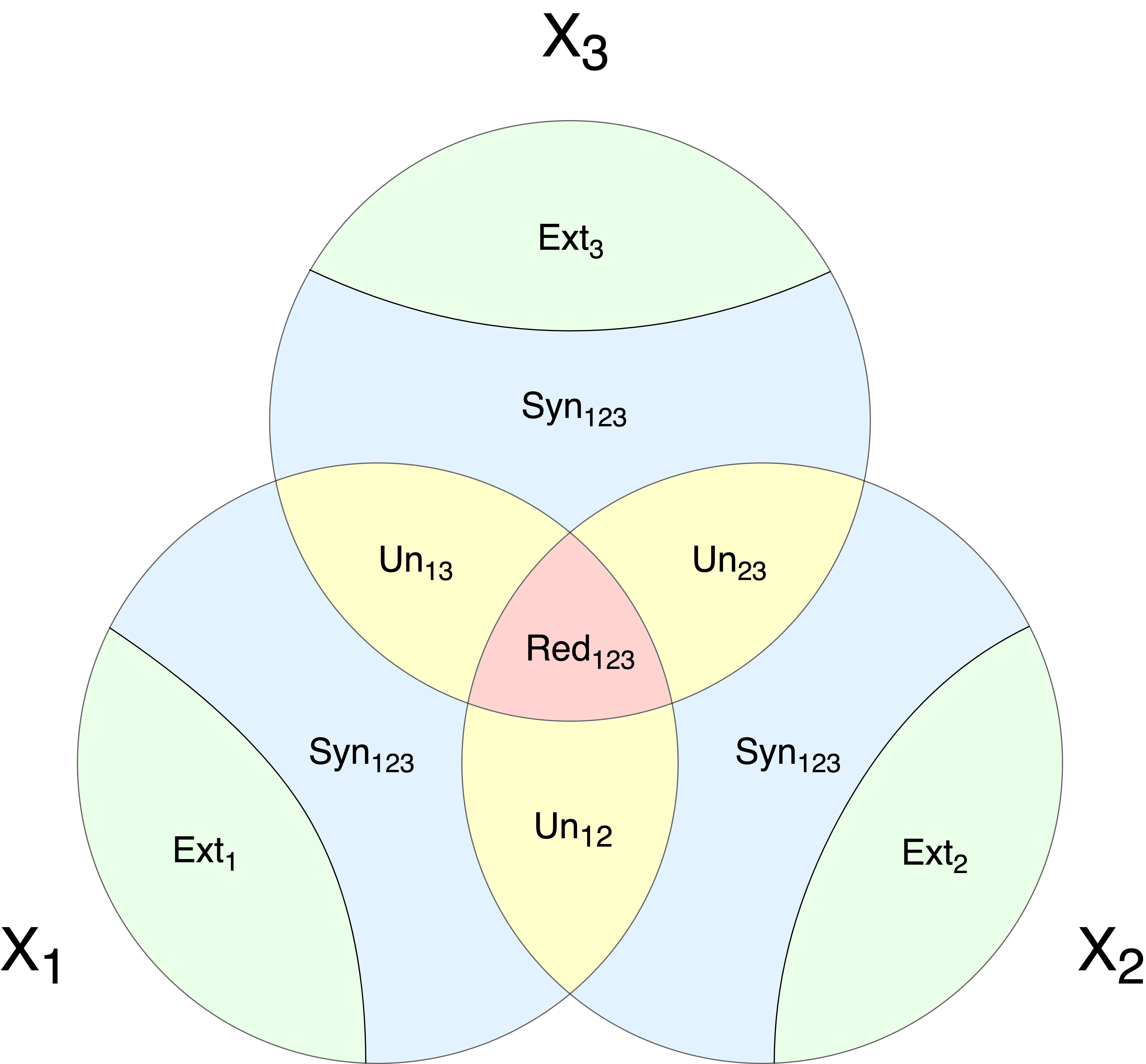}}
\caption{Venn diagram of SID's Formal Version.}
\label{fig:Venn_2.0}
\end{figure}

\begin{definition} [System Information Decomposition Framework]
SID is a conceptual system decomposition framework based on information entropy, that can divide the whole information entropy of a multivariate system into non-overlapping information atoms according to the relationship among variables. In this framework, redundant information represents the common or overlapping information of all the variables; unique information represents information that is only owned by two variables but not by others; and synergistic information represents the information that can be known from any variable only when the other variables are observed simultaneously. 
\end{definition} 

In the SID framework, the Venn diagram unifies the Shannon's framework\ref{sec:Information Theory Framework} and PID framework \ref{sec:Partial Information Decomposition Framework}. For a intuitive presentation, we only give the Venn diagram of three-variable system ($\{X_1, X_2, X_3\}$) in this paper. Since Venn diagrams become difficult to understand as the number of variables increases, for the visualization of general systems' SID, we will explore other visualization tools for SID in the Discussion.

\subsection{SID and Information Measure}
\label{sec:SID and Existing Information Measure}

In addition to Axiom \ref{axiom:Quantitative Computation} and Definition \ref{definition:External Information} for the relationship between SID and mutual information, conditional entropy and joint conditional entropy, there are still some important information measures that deserve our attention.

\begin{Corollary} [Joint Entropy Decomposition]
\label{Corollary:Joint Entropy Decomposition}
    For any subsystem with 3 variables, $H (X_{1}, X_{2}, X_{3}) = Ext(X_{1}) + Ext(X_{2}) + Ext(X_{3}) + Un(X_{1},X_{2}) + Un(X_{1},X_{3}) + Un(X_{2},X_{3}) + 2 * Syn(X_{1}, X_{2}, X_{3}) +  Red (X_{1}, X_{2}, X_{3})$.
\end{Corollary}

Based on Corollary \ref{Corollary:Joint Entropy Decomposition}, which can be easily proved by Axiom \ref{axiom:Quantitative Computation}, we can have a deeper understanding of information atoms, that is, any information atom can be understood as some kind of information stored by $m$ variables, and at least $n$ variables need to be known to obtain the information ($m>n, m, n \in \mathbb{Z}$). Specifically, the external information of the system is owned by the variable independently, so $m=1$ and $n=1$; redundant information is owned by all variables, so $m = number\>of\>variables$ and $n=1$; unique information is owned by two variables, Therefore $m=2$ and $n=1$; synergistic information is shared by all variables, so $m = number\>of\>variables$ and $n = number\>of\>variables - 1$. Therefore, the joint entropy decomposition is the sum of each information atom multiplied by its $m-n$ quantity. This perspective will deepen our understanding of the essence of information atoms and facilitate our exploration of the joint entropy decomposition of systems with more than three variables. Besides, this phenomenon also reflects the differences between information measures and Venn diagrams. Considering that Venn diagrams cannot fully reflect the nature of information decomposition, alternative visualization solution will be discussed in the discussion section.


\begin{Corollary} [Total Correlation Decomposition]
\label{Corollary:Total Correlation Decomposition}
    For any subsystem with 3 variables: 
    \begin{align} TC (X_{1}, X_{2}, X_{3}) =  Un(X_{1},X_{2}) + Un(X_{1},X_{3}) + Un(X_{2},X_{3}) + Syn(X_{1}, X_{2}, X_{3}) +  2 * Red (X_{1}, X_{2}, X_{3})
    \end{align}
\end{Corollary}

\begin{Corollary} [Intersection Information Decomposition]
For any system with 3 variables, its Intersection Information:
\begin{align} CoI (X_{1}, X_{2}, X_{3}) = Red(X_{1}, X_{2}, X_{3}) - Syn(X_{1}, X_{2}, X_{3})\end{align}
\end{Corollary}
According to the calculation of $CoI (X,Y, Z) = H (X_{1}, X_{2}, X_{3}) + H (X_{1}) + H (X_{2}) + H (X_{3}) - H (X_{1}, X_{2}) - H (X_{1}, X_{3}) - H (X_{2}, X_{3})$, $Col$ is symmetry and unique for a system, which also verifies if one of Synergistic or Redundant information is symmetric, then so is the other.

\section{Case Studies}
\label{sec:Case Studies}

In this section, through a series of case analyses, we elucidate the unique properties of the SID framework and its capacity to uncover higher-order relationships that surpass the capabilities of current information and probability measures.

Without loss of generality, we can construct a case that includes both macro and micro perspectives, which can not only analyze the properties of SID at the macro level but also obtain "ground truth" through known micro properties. First, we construct six uniformly distributed Boolean variables $a,b,c,d,e,f$, ensuring that these variables are independent. We then create new variables by performing XOR operations on the existing variables: let $g = c\oplus e$, $h = d\oplus f$, $i = c\oplus f$, and $j = d\oplus e$, where $\oplus$ represents XOR.

Next, we construct new macro variables by combining these micro variables: let $X_{1} = (a,b,c,d)$, $X_{2} = (a,b,e,f)$, $X_{3} = (c,d,e,f)$, $X_{4} = (a,c,e,h)$, $X_{5} = (a,b,g,h)$, $X_{6} = (a,b,i,j)$. The combination method involves simple splicing; e.g., when $a=1$, $b=0$, $c=1$, $d=1$, $X_{1}$ is equal to $1011$. Appendix \ref{Appendix} provides a concrete example that matches this design. As the micro-level variables are independent of each other, this combination ensures that the properties of the macro variables are a combination of the properties of the micro variables.

Then, we fix $X_1$ and $X_2$ as constants and form different three-variable systems (Cases 1-4) by adding $X_3$, $X_4$, $X_5$, and $X_6$ respectively, as shown in Table \ref{tab:case detail}. After knowing the microscopic dynamics of these cases, we can more intuitively analyze their characteristics under the SID framework.

\begin{table}[h]
\centering
\begin{tabular}{|c|c|c|c|}
\hline
Case & Variables & Micro-component & Micro-relationship\\
\hline
 1  & $X_1$, $X_2$, $X_3$ & $abcd$, $abef$, $cdef$ & $abcdef$ are independent \\
\hline
 2  & $X_1$, $X_2$, $X_4$ & $abcd$, $abef$, $aceh$ & $abcdef$ are independent, $h = d\oplus f$ \\
\hline
 3  & $X_1$, $X_2$, $X_5$ & $abcd$, $abef$, $abgh$ & $abcdef$ are independent, $g = c\oplus e$, $h = d\oplus f$\\
\hline
 4  & $X_1$, $X_2$, $X_6$ & $abcd$, $abef$, $abij$ & $abcdef$ are independent, $i = c\oplus f$, and $j = d\oplus e$\\
\hline
\end{tabular}
\caption{Cases Construction.}
\label{tab:case detail}
\end{table}

It is worth noting that these four cases yield identical results under existing non-information-decomposition measures. The system has 64 equally probable outcomes, each variable has 16 equally probable outcomes, the total information amount in the system is 6, the pairwise mutual information between variables is 2, and the conditional entropy is 2. Existing non-information-decomposition methods cannot identify the differences observed in these four examples, and information-decomposition-related methods cannot decompose the information entropy of all variables in the system and identify the symmetric relationship among variables.

However, the four systems exhibit three distinct internal characteristics under the SID framework. Since these examples comprise mutually independent micro variables, we can intuitively map the micro variables to the information atoms in each case. In Case 1, the micro variables $a,b$ provide 2-bit unique information between $X_1$ and $X_2$ ($c,d$ correspond to $X_1$ and $X_3$, $e,f$ correspond to $X_2$ and $X_3$). In Case 2, micro variable $a$ provides 1-bit redundant information, while $b$, $c$, and $e$ provide 1-bit unique information between $X_1$ and $X_2$, $X_1$ and $X_4$, $X_2$ and $X_4$ respectively. The XOR relationship between $d-f-h$ provides 1-bit synergistic information between variables. In Cases 3 and 4, micro variables $a$ and $b$ provide 2-bit redundant information, and XOR relationships of $c-e-g$, $d-f-h$, and $c-f-i$, $d-e-j$ provide 2-bit synergistic information for the two cases, respectively. Figure \ref{fig:SID Venn Diagrams} displays the SID Venn diagrams for Cases 1–4. 

\begin{figure}[htbp]
    \centering
    \begin{subfigure}{0.45\textwidth}
        \centering
        \includegraphics[width=\linewidth]{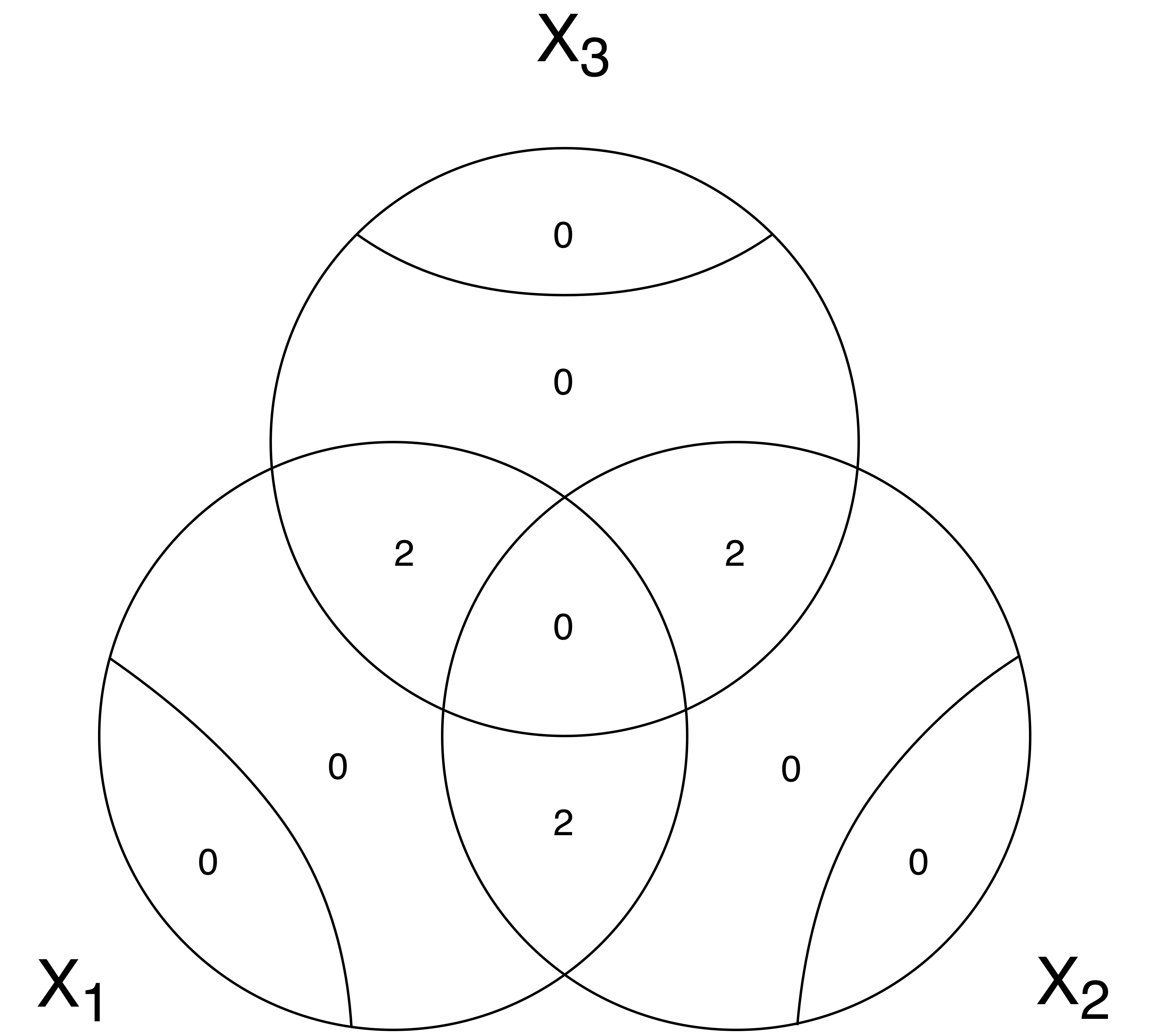}
        \caption{Case 1.}
    \end{subfigure}
    \hfill
    \begin{subfigure}{0.45\textwidth}
        \centering
        \includegraphics[width=\linewidth]{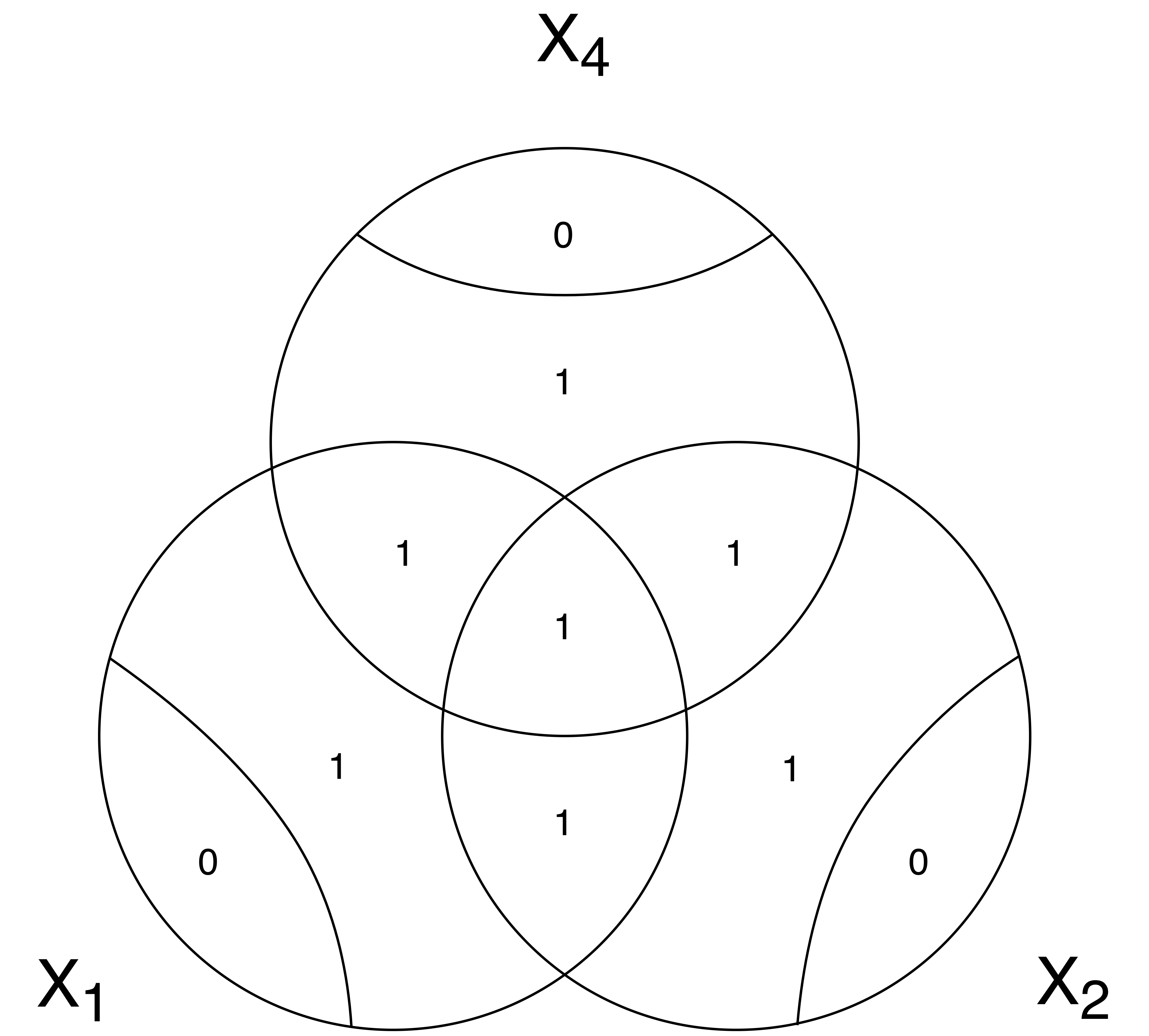}
        \caption{Case 2.}
    \end{subfigure}
    \newline
    \begin{subfigure}{0.45\textwidth}
        \centering
        \includegraphics[width=\linewidth]{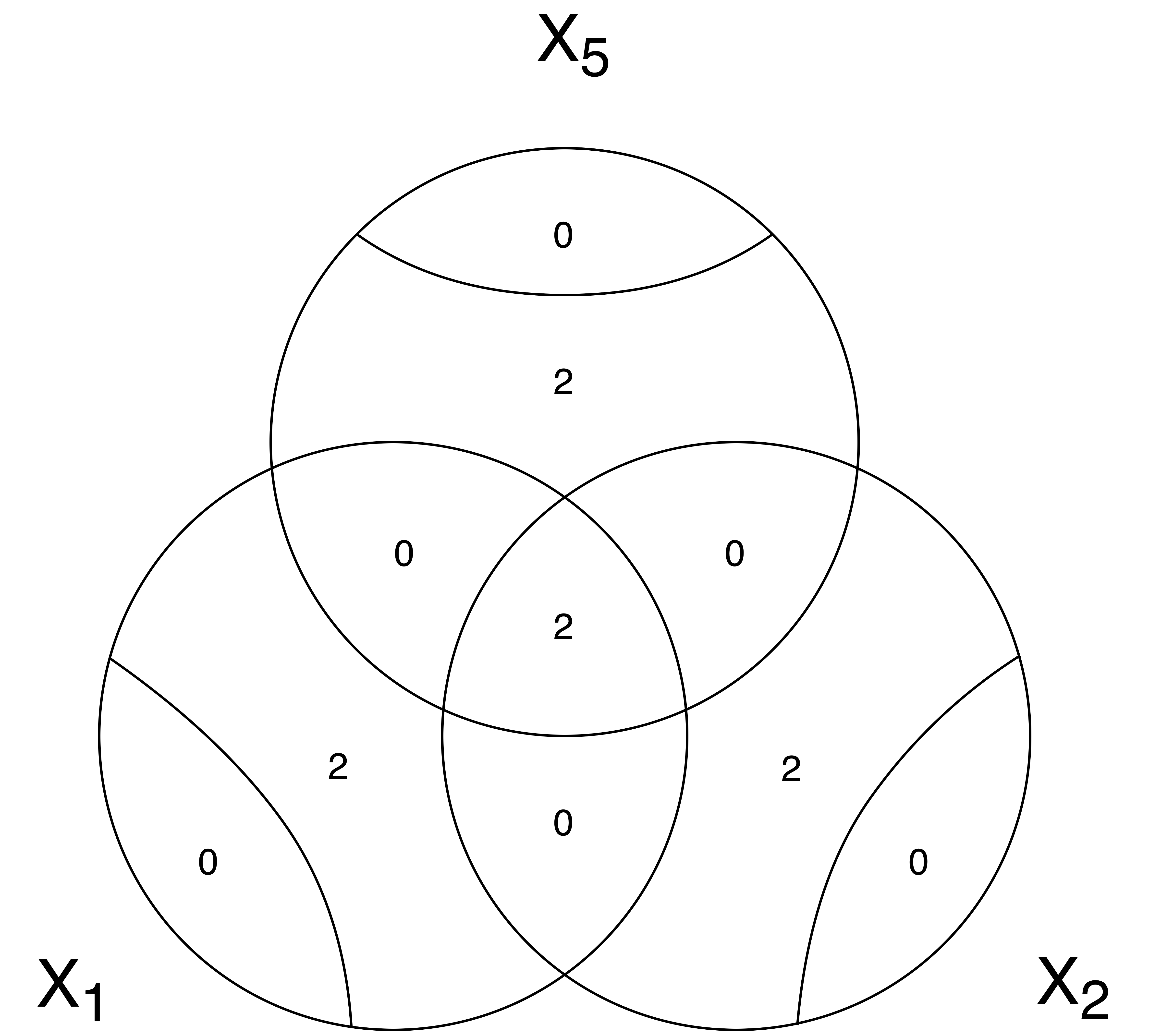}
        \caption{Case 3.}
    \end{subfigure}
    \hfill
    \begin{subfigure}{0.45\textwidth}
        \centering
        \includegraphics[width=\linewidth]{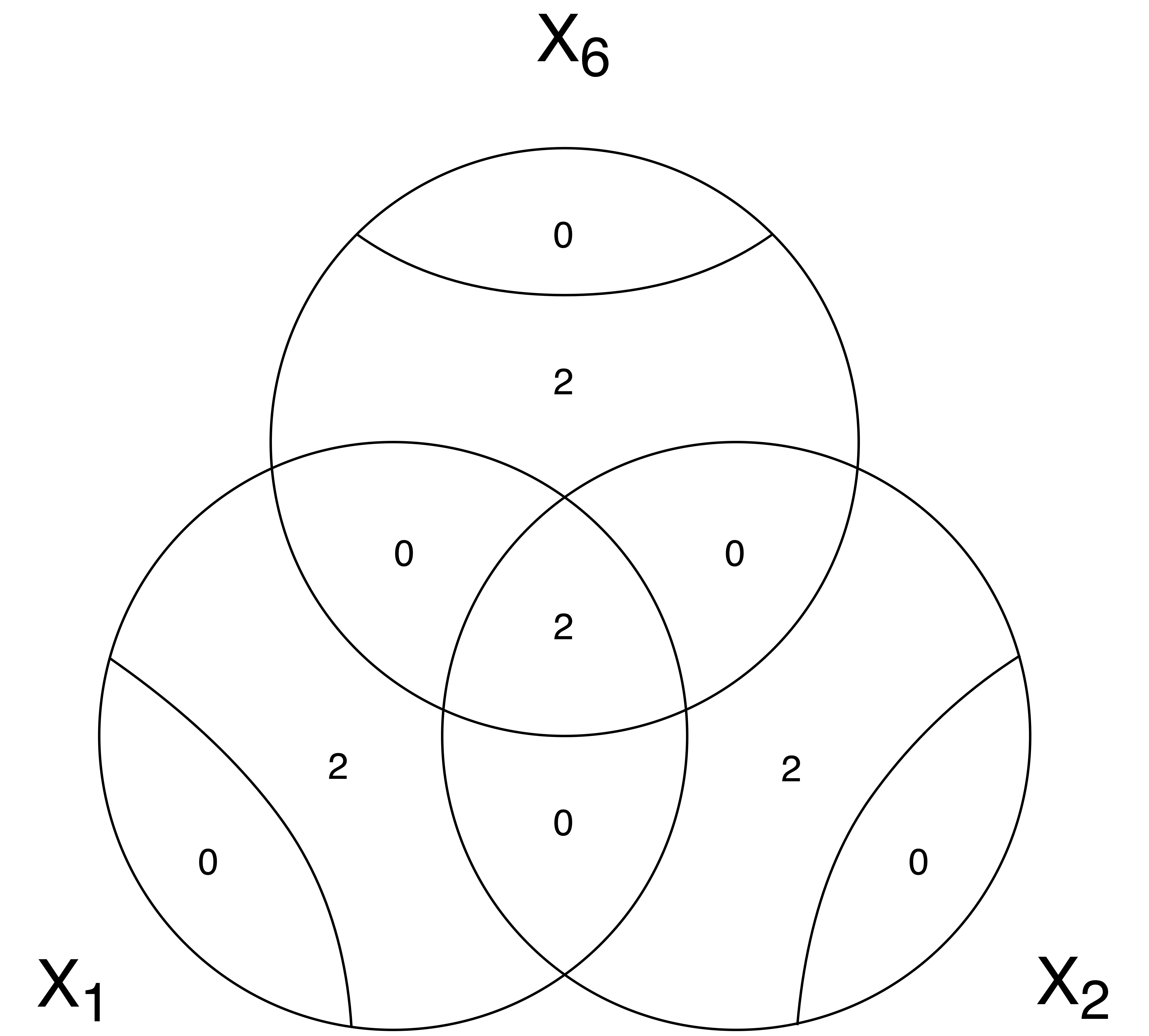}
        \caption{Case 4.}
    \end{subfigure}
    \caption{SID Venn Diagrams for Cases 1-4. In Case 1, $a,b$ provide 2-bit unique information between $X_1$ and $X_2$ ($c,d$ correspond to $X_1$ and $X_3$, $e,f$ correspond to $X_2$ and $X_3$). In Case 2, $a$ provides 1-bit redundant information, $b$, $c$, and $e$ provide 1-bit unique information between $X_1$ and $X_2$, $X_1$ and $X_4$, $X_2$ and $X_4$ respectively. The XOR relationship between $d-f-h$ provides 1-bit synergistic information. In Cases 3 and 4, $a, b$ provide 2-bit redundant information, XOR relationships of $c-e-g$, $d-f-h$, and $c-f-i$, $d-e-j$ provide 2-bit synergistic information for the two cases, respectively.}
    \label{fig:SID Venn Diagrams}
\end{figure}

\section{Calculation of SID}
\label{sec:Calculation of SID}
Although we have proposed the framework of SID and proved the symmetry of information atoms, the problem of exact computation has not been fully resolved. Therefore, in this section, we alternatively propose the desired properties that the calculation method of the SID framework should satisfy, and accept any method that can meet these properties. Additionally, we propose a direct methods for some special cases and two novel methods for more general cases and validate their accuracy and applicability through the examination of the cases \ref{sec:Case Studies}.

\subsection{Properties of Calculation Methods for SID}
\begin{Property}[Compatible with Shannon's formula]
\label{property:Shannon's formula}
The sum of certain information atoms should equal to the mutual information and conditional information. For a three-variable system, it is Axiom \ref{axiom:Quantitative Computation}.
\end{Property}
The information atoms can be regarded as a finer-grained division of Shannon’s information entropy calculation, so calculation methods such as information entropy, mutual information, and conditional entropy can accurately calculate the sum of some information atoms, which means that the SID’s calculation should conform to the Shannon formula. It is worth noting that when the specific PID calculation method calculates the value of one information atom, the rest of the information atoms will also get the results according to Axiom \ref{axiom:Quantitative Computation}. This means that the calculation method of SID only needs to focus on one information atom in the system. 

\begin{Property}[Computational Symmetry]
\label{property:Computational Symmetry}
The results of SID calculation should satisfy Theorems \ref{Theorem:Symmetry of Redundant Information}, \ref{Theorem:Symmetry of Unique Information}, and \ref{Theorem:Symmetry of Synergistic Information}. 
\end{Property}

For the same system, the order of variables in the calculation method will not affect the results. This ensures that the SID framework provides a consistent decomposition of information entropy, regardless of the order of variables. Specifically, for redundant information and synergistic information, changing the order of any variable in the calculation method will not change the result; for unique information, exchanging the positions of the two focused variables or changing the order of the remaining variables will not change the result.

\begin{Property}[Non-negativity of information atoms]
\label{property:Non-negativity of information atoms}
After applying SID, the value of any information atom is greater than or equal to zero. This non-negativity property holds because information measures, the degree of uncertainty are always non-negative as per the principles of information theory.
\end{Property}

Although the computational problem of information atoms has not been completely solved yet, just like finding the Lyapunov function, for a specific case, we can often use specific methods, analysis, and some intuition to get the result. For example, a direct and rigorous method is to use properties \ref{property:Shannon's formula} and \ref{property:Non-negativity of information atoms}. 

\begin{Proposition} [Direct Method]
\label{Proposition:Direct Method}
If certain mutual information or conditional entropy is zero, we can directly draw the conclusion that: (1) the redundant information and the corresponding unique information are zero if some mutual information is zero, or (2) the synergistic information and the corresponding unique information are zero if some conditional entropy is zero. Then, we can obtain the values of the remaining information atoms.
\end{Proposition} 

For a more general scenario, we are going to give a calculation formula that can be applied to most situations and a neural network method that can give approximate values.

\subsection{A Calculation Formula}

Although we can calculate some cases through the Direct Method \ref{Proposition:Direct Method} or from the perspective of construction like previous analysis of case \ref{sec:Case Studies}, we need to find a general solution to make the SID framework applicable in a wider range of scenarios. After analyzing a large number of construction cases with ground truth, we reveal the correspondence between the values of information atoms and certain structures within the data, which we called it Synergistic Block and Unique Block.

\begin{definition}[Synergistic Block and Unique Block]

For a joint probability distribution table of three random variables $X_1,X_2,X_3(\mathcal{X}_1,\mathcal{X}_2,\mathcal{X}_3)$, given $X_{1}=x_{1}$, the possible values of $X_2$ and $X_3$ form a set $\mathcal{J}(x_1) = \{x_2:P(X_{2}=x_2|X_{1}=x_{1})\ne0\} $, $\mathcal{K}(x_1) = \{x_3:P(X_{3}=x_3|X_{1}=x_{1})\ne0\}$. 
The Synergistic Block of $X_2, X_3$ given $X_1 = x_1$ is the events such that $x_1,x_2,x_3 \in (\mathcal{X}_1/x_1) \times \mathcal{J}(x_1) \times \mathcal{K}(x_1)$. The Unique Block of $X_2$ given $X_{1}=x_{1}$ is the events such that $x_1,x_2,x_3 \in (\mathcal{X}_1/x_1) \times \mathcal{J}(x_1) \times (\mathcal{X}_3/\mathcal{K}(x_1))$ (the events $x_1,x_2,x_3 \in (\mathcal{X}_1/x_1) \times (\mathcal{X}_2 / \mathcal{J}(x_1)) \times \mathcal{K}(x_1)$ for $X_3$).
\end{definition}

Take Table \ref{Table:Case Table} as an example, we fixed the value of $X_{1}=0000$, and marked the values of all variables in this scenario in yellow. Then, we mark the values where $X_{2}$ to $X_{6}$ still take the same value when $X_{1} \ne 0000$ as pink. Taking $X_{1}$, $X_{2}$ and $X_{4}$ as examples, we marked the synergistic blocks in \textbf{bold}, and marked the unique blocks of $X_{2}$ and $X_{3}$ in \textit{italics}. Besides, although not as obvious as the previous two, redundant information also has corresponding redundant blocks. Based on this observation, we propose a conjecture for unique information and synergistic information.

\begin{Conjecture} [Information Atom Identification]
\label{Conjecture:Information Atoms Identification}
For a three-variable system, the synergistic information is greater than zero if and only if the synergistic block exists. 
The unique information between two variables is greater than zero if and only if fix any of them, the remaining variable have unique block. 
\end{Conjecture}

Based on the Conjecture \ref{Conjecture:Information Atoms Identification}, we propose a calculation formula for synergistic information (Unique information is also similar). Although this method is currently only applicable to three-variables cases constructed using uncorrelated features of microscopic variables like case \ref{sec:Case Studies}, it provides inspiration for the exploration of a general computational methods. The specific calculation method for synergistic information for a three-variable system involving $X_1$, $X_2$, and $X_3$ is as follows:

\begin{align}
&Syn(X_{1}, X_{2}, X_{3}) \nonumber \\
&= \sum_{x_{1},x_{2},x_{3}} P(x_{1},x_{2},x_{3}) *\log ( \frac{\sum_k P(X_{2} = x_{2}, X_{3} = k)}{P(X_{2} = x_{2},X_{1} = x_{1})} *
\frac{\sum_j P(X_{3} = x_{3}, X_{2} = j)}{P(X_{3} =  x_{3}, X_{1} = x_{1})} * \frac{P(X_{1} = x_{1})}{P(X_{2} = j, X_{3} = k)}) \nonumber \\
&- H(X_{1}|X_{2},X_{3}) \nonumber \\
& \forall (x_1,x_2,x_3), j \in \mathcal{J}(x_1) = \{x_2:P(X_{2}=x_2|X_{1}=x_{1})\ne0\}, k \in \mathcal{K}(x_1) = \{x_3:P(X_{3}=x_3|X_{1}=x_{1})\ne0\}
\end{align}

In the previous case \ref{sec:Case Studies}, since the data is relatively uniform, fixing any value of $X_{1}$ will have the same result, so we can quickly show the synergistic information of the four cases by fixing $X_{1}=0000$. In these cases, the $\log$ part of the formula can be intuitively understood as $\log$ (yellow + synergistic block / yellow), which is $\log (4/4) = 0$ in case 1; $\log (8/4) = 1$ in case 2; $\log(16/4)=2$ in cases 3 and 4. Unique information can also be calculated by a similar method like $\log $(yellow + unique block / yellow). 

\subsection{An Approximate Method by Neural Information Squeezer}
Another possible method is to use a generalized form of neural information squeezer (NIS, a machine learning framework by using invertible neural networks proposed in Ref \cite{zhang2023nis}) to numerically calculate the redundancy of the system, and then to derive other information atoms. 

\begin{figure}[htbp]
\centering
\includegraphics[width=1\linewidth]{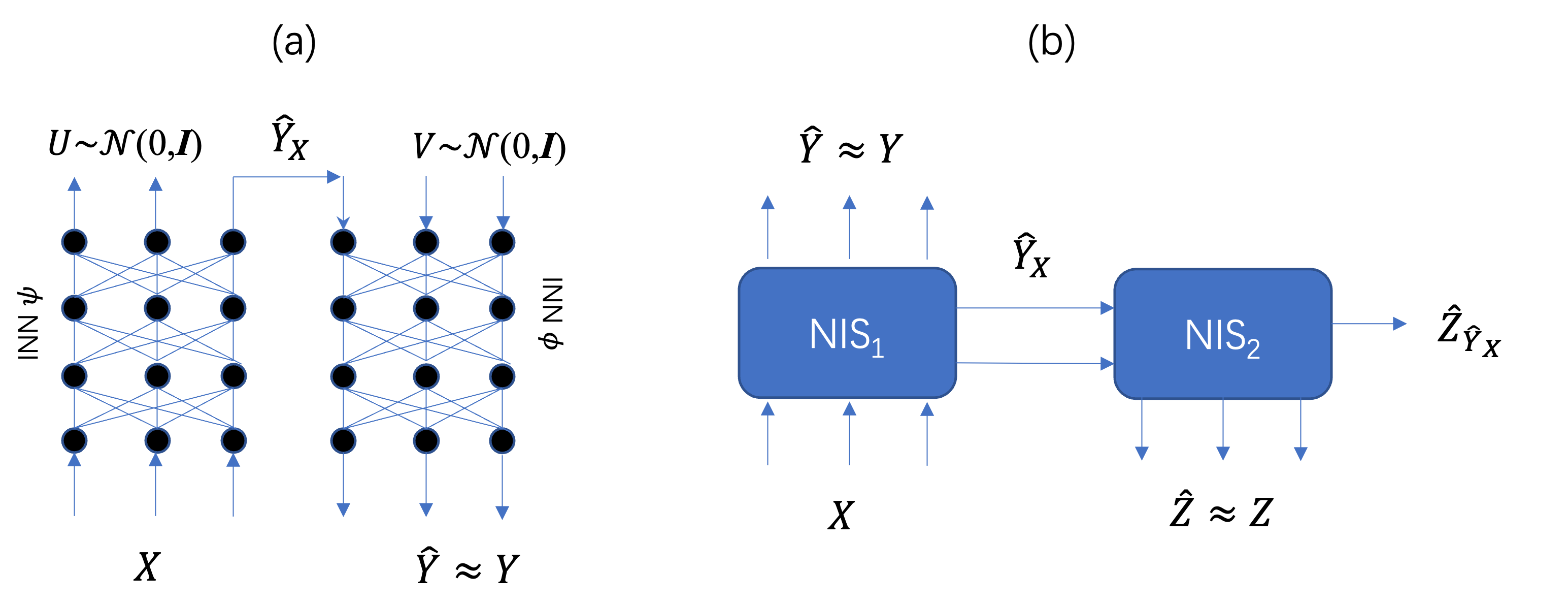}
\caption{A generalized form of the Neural Information Squeezer network (NIS, see \cite{zhang2023nis}) to calculate mutual information(a) and redundancy(b) of a trivariate system ($X,Y,Z$). In (a), there are two invertible neural networks ($\psi,\phi$) which can play the roles of encoder and decoder, respectively. The whole network accepts the input $X$ to predict $Y$, and the intermediate variable $\hat{Y}_X$, which is the minimum low-dimensional representation of $X$, can be used to calculate the mutual information between $X$ and $Y$. In (b), two NIS networks are stacked together. The first one is just the network in (a), and the intermediate variable $\hat{Y}_X$ is fed into the second NIS network to predict $Z$. Then the intermediate variable, $\hat{Z}_{\hat{Y}_X}$ which is the minimum low-dimensional representation of $\hat{Y}_X$, can be used to calculate the redundancy of the system $\{X,Y,Z\}$.}
\label{fig:nis}
\end{figure}

As shown in Figure \ref{fig:nis}(a), the NIS framework has two parts: an encoder and a decoder. The encoder can accept any real vector variable with dimension $p$, and it contains two operators: a bijector $\psi$ modeled by an invertible neural network (see details in \cite{zhang2023nis}) with dimension $p$ and a projector $\chi$ which can drop out the last $p-q$ dimensions from the variable $\psi_p(X)$ to form variable $U$. The remaining part ($\hat{Y}_X$) can be regarded as a low-dimensional representation of $X$ which will be used to construct the target $Y$ via another invertible neural network $\phi$ by mapping $[V, \hat{Y}_X]$ into $\hat{Y}$, where $V\sim \mathcal{N}(0,I)$ is a $p'-q$ dimensional random noise with Gaussian distribution, where $p'$ is the dimension of $Y$. Then, we need to train the whole framework to conform that (1) $\hat{Y}$ approximates the target variable $Y$, and (2) $U$ follows a $p-q$ dimensional standard normal distribution. It can be proven that the following proposition holds:

\begin{Proposition}\label{Proposition:NIS_MI}For any random variables $X$ with $p$ dimension and $Y$ with $p'$ dimension, and suppose $p$ and $p'$ are very large, then we can use the framework of Figure \ref{fig:nis}(a) to predict $Y$ by squeezing the information channel of $\hat{Y}_X$ as the minimum dimension $q^*$ but satisfying $\hat{Y}\approx Y$ and $U\sim\mathcal{N}(0,I)$. Further, if we suppose $H(X)>H(X|Y)>0$, then:
\begin{equation}
    \label{eqn.entropy_ins}
    H(\hat{Y}_X)\approx I(X;Y),
\end{equation}
and
\begin{equation}
    \label{eqn.hu_ins}
    H(U)\approx H(X|Y).
\end{equation}

\end{Proposition} 

We will provide the proof in the appendix. The reason why we require that the numbers of dimensions of $X,Y$ are large is because the maximal $q$ for accurate predictions may not be integer if $p,p'$ are small. Therefore, we can enlarge the dimensions by duplicating the vectors.

To calculate the redundancy for a system with three variables: $X,Y,Z$, we can use the NIS network twice, as shown in Figure \ref{fig:nis}(b). The first NIS network is to use the intermediate variable $\hat{Y}_X$, the dense low-dimensional representation of $X$ with the minimum dimension $q$, to construct $Y$. Then, the second NIS network is to use $\hat{Z}_{\hat{Y}_X}$, the minimal dimensional dense low-dimensional representation of $\hat{Y}_X$ to construct $Z$. After these two steps, the Shannon entropy of the intermediate variable of $NIS_2$: $\hat{Z}_{\hat{Y}_X}$ can approach the redundancy. Thus, the redundancy of the system can be calculated approximately in the following way:

\begin{equation}
    \label{eqn.nis_red}
    Red(X,Y,Z)\approx H(\hat{Z}_{\hat{Y}_X}).
\end{equation}

To verify that $Red(X,Y,Z)$ calculated in this way can be regarded as the redundancy of the system, we need to prove that Equation \ref{eqn.nis_red} satisfies the property of symmetry for all the permutations of $X,Y,Z$, i.e., the following proposition:

\begin{Proposition}\label{Proposition:NIS_RED}For a system with three random variables $X,Y,Z$, without losing generality, we suppose that the conditional information satisfy $H(X)>H(X|Y)>0$ , $H(X)>H(X|Z)>0$, and $H(Y)>H(Y|X)>0$, then the redundancy calculated in Equation \ref{eqn.nis_red} is symmetric:
\begin{equation}
    \label{eqn.redundency_NIS}
    Red(X,Y,Z)\approx Red(X,Z,Y).
    \end{equation}
\end{Proposition}

To be noticed that $Red(X,Z,Y)\approx H(\hat{Y}_{\hat{Z}_X})$ is different from $Red(X,Y,Z)$ in the way that the order of the predictions from $X$ is $Z$ and then $Y$.

The proof of Theorem \ref{Proposition:NIS_RED} is also provided in the appendix. With the calculation of redundancy, we can easily calculate unique and synergistic information atoms. Furthermore, we can extend the method to systems with more variables by just stacking more NIS networks in the similar way as shown in Figure \ref{fig:nis} (b).

However, there are two disadvantages to this method, one is that the calculation is inaccurate and requires a large number of training epochs. Second, the numbers of dimensions of all variables must be large enough such that the independent information among the variables can be discarded by dropping out the dimensions. Further studies are needed.

To verify that the NIS framework can calculate redundant information, we conducted numerical experiments using Case3 as an example, as Figure\ref{fig:nis calculation} shows, where the mutual information between each pair of variables and the redundant information is 2 bits. 

In this experiment, variable $X_1$ is used as the input of NIS1 in the framework, with $X_2$ predicted as the target $Y$, and the intermediate variable $\hat{Y}_X$ is fed into NIS2 to predict $X_3$. In this experiment, both inputs and targets were expanded to 64 dimensions by direct replication of the original variables, and let the two intermediate variables in the NIS maintain consistent dimensions, both denoted by $q$ . The minimum dimension of $\hat{Y}_X$ and $\hat{Z}_{\hat{Y}_X}$ are selected by monitoring the changes in the loss curves. 

From the above results, it can be seen that when $q$, the dimension of intermediate variable, is relatively large, the entropy of the intermediate variable is quite accurate for mutual information or redundant information. As the $q$ drops below a threshold, the loss signally increases, indicating that the intermediate variable cannot capture all the mutual information and the redundant information.

\begin{figure}[htbp]
    \centering
    \includegraphics[width=1\linewidth]{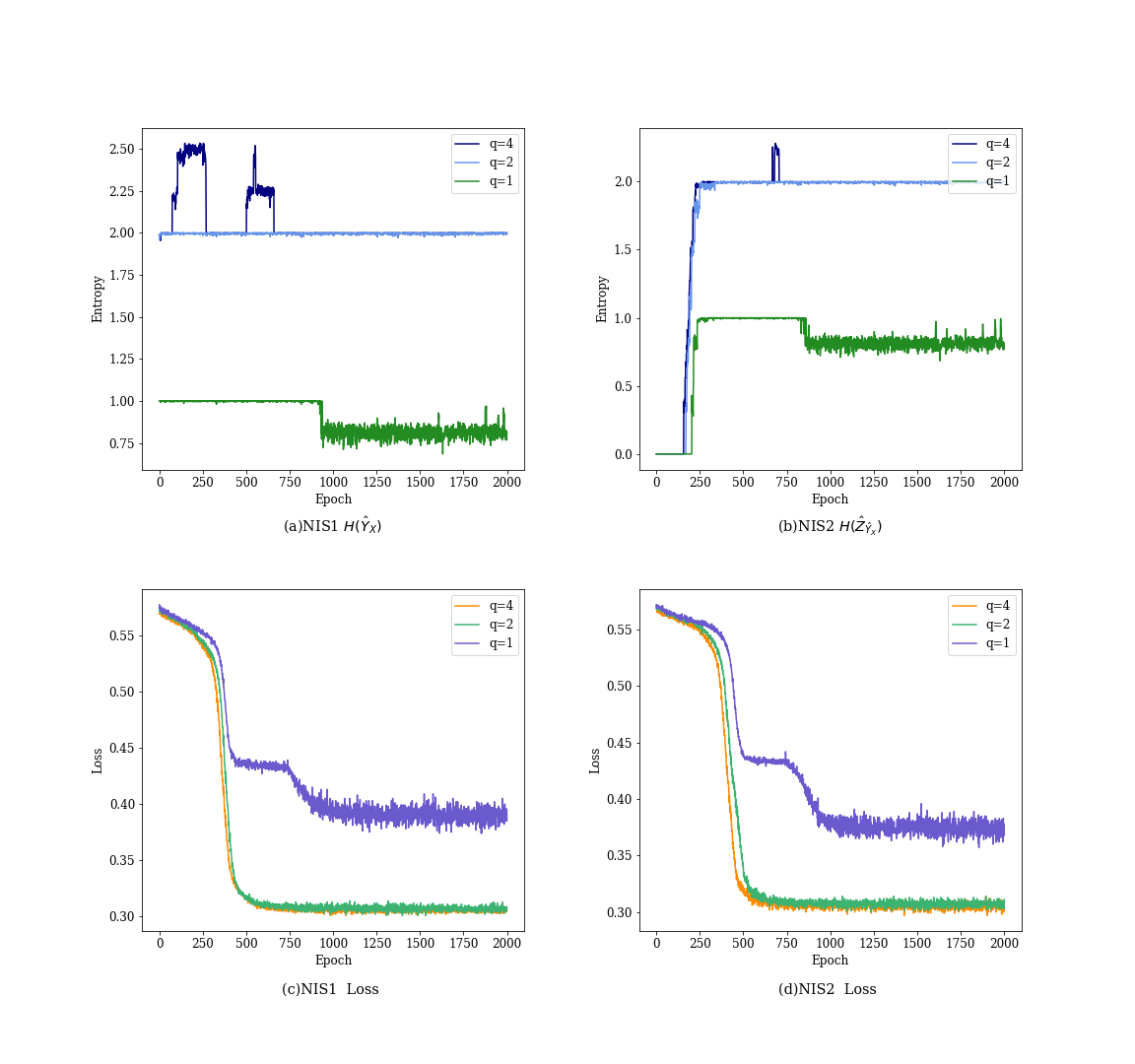}
    \caption{(a) The changes of $H(\hat{Y}_X)$ in NIS1 under $q=4,2,1$ respectively; (b) The changes of $H(\hat{Z}_{\hat{Y}_X})$ in NIS2 under $q=4,2,1$ respectively; (c) The changes of training loss in NIS1 under $q=4,2,1$ respectively; (d) The changes of training loss in NIS2 under $q=4,2,1$ respectively. The same experiments were conducted for the other three cases, and the redundant information could be accurately calculated under the NIS framework.}
    \label{fig:nis calculation}
\end{figure}

To sum up, NIS can compute the system's redundancy numerically, which subsequently allows us to extract additional information atoms. While the NIS, in its current state as a continuous computing framework, grapples with precision issues when resolving discrete problems, it presents a promising pathway for quantifying information atoms. As part of our ongoing research, we aim to enhance the NIS's precision in tackling discrete problems and explore the problems into general continuous variables.

\section{Discussion}
\label{sec:Discussion}

\subsection{SID and PID}
As an information decomposition method compatible with PID's conceptual framework, SID expanded the scope of information decomposed from the mutual information of the source variables and the target variable to the whole information of the system, and shows the symmetry of information atoms among variables. It is worth noting that SID is not based on any existing PID calculation methods, but a set of computational properties that should be satisfied.

Based on these expansion, the biggest difference between SID and PID is the analysis perspective: PID focuses on the directed pairwise (second-order) relationship between the set of source variables and the target variable, while SID focuses on all relationship among variables in the system, from pairwise to undirected high-order relationship. This exhaustion of relationships enables SID to pay attention to the relationships among the source variables and the high-order symmetric relationships that PID ignores. Take Case \ref{sec:Case Studies} as an example. From the perspective of PID, there are directional redundant, synergistic and unique information from $X_1$ and $X_2$ to the target variable, but the information interaction relationship between $X_1$ and $X_2$ is unknown. Also, it cannot be realized from PID that the synergistic information provided by $X_1, X_2$ to the target variable is only a partial understanding of the undirected synergistic effect among the three variables, and this effect also occurs when $X1$ or $X2$ as target. In addition, on the basis of being compatible with PID, SID adds more constraints, such as Theorem \ref{Theorem:Symmetry of Redundant Information} \ref{Theorem:Symmetry of Unique Information} \ref{Theorem:Symmetry of Synergistic Information}, so it provides more ways in the calculation of information decomposition. For example, in Proposition \ref{Proposition:Direct Method}, it can be inferred that the redundant information is zero through the presence of variable pairs with zero mutual information in the variable set, which is not satisfied in some existing PID calculation methods. 

To sum up, SID extents the analysis scope and reveals several essential natures of information atoms on the basis of being compatible with the PID framework, and greatly expands the application scenarios of information decomposition, which will be discussed in the next few paragraphs.

\subsection{SID and Higher-order Measurement}
The holism-versus-reductionism debate persists in modern literature \cite{mitchell2009complexity}. Those who hold a reductionist view believe that any system can be divided into many subsystems, and we can fully understand the entire system by studying the properties of the subsystems and their connections, which is also the research philosophy followed by most disciplines \cite{gallagher1999beyond}. But holism holds that the system should be treated as a whole because the splitting of the system will inevitably lose the understanding of some of its properties \cite{smuts1926holism}. This contradiction seems irreconcilable when we don't discuss in detail how to decompose the system.

However, the SID offers a perspective that can explain this conflict by accounting for higher-order relationships in the system that are not captured by previous measures. To better divide the different measures, we divide information entropy into first-order measures, which reflect a certain attribute of a single variable. Mutual information and conditional entropy, on the other hand, can be divided into second-order measures, which capture some aspects of pairwise relationships between variables \cite{cover1999elements}. Although among the second-order measurement, information theory's cross-entropy can measure the information shared among multiple variables, it still captures linear superpositions of second-order relationships, which provides limited insight into multivariate interactions. But under the SID framework, redundant, synergistic, and unique information can be regarded as three- or higher-order measures, revealing a new dimension of multivariate relationships that is entirely distinct from the first and second orders and facilitating a deeper comprehension of higher-order system relationships. In the case analysis, the internal structure of Case 1 aligns well with the results of the second-order measures, and can be considered a reducible, decomposable system. Cases 2, 3, and 4, however, have internal structures that cannot be captured by second-order measures and are thus regarded by holism as systems that cannot be decomposed and understood individually. To some extent, SID and the case analysis offer an explanation that bridges the gap between holism and reductionism; that is, some of the system properties that holism insists cannot be understood separately might be explained by higher-order measures or decomposition methods.

\subsection{Potential Application}
In addition to philosophical discussions, higher-order measures can also be applied to many fields. A foreseeable application across many domains comes from that SID deepens our understanding of data, measures, and information. In the case studies \ref{sec:Case Studies}, the data contains information about the construction of the four variable systems, but the inner relationship of the system cannot be captured by probability measures or existing information measures. That means the incompleteness of measures may limit our ability to analyze existing systems, even if we have obtained complete data. Therefore, conducting higher-order information measures in the analysis of complex systems may offer valuable insights, especially in the field where traditional information measures fail to capture the relationship among systems. A worth exploring direction is the quantitative analysis Higher-order Networks \cite{bianconi2021higher}. Since SID can provide a data-driven framework for identifying and analyzing of high-order network structures, it may potentially impact the analysis and understanding of complex systems across various domains \cite{battiston2021physics}. For example, in studying neural networks and brain connectivity \cite{bullmore2009complex}, the SID framework can provide further insights into the higher-order information flow between multiple neurons or brain regions, which will allow us to directly generate higher-order network models between neurons through the temporal data of multiple neurons, and use this model to explain the implementation of specific functions; in ecological \cite{levin2005self}, financial, or social systems, the quantitative characterization of high-order relationships among multiple agents can assist in the development of more accurate models and forecasts, as well as the design of effective control methods. Also, this combination is also a two-way promotion: Since Venn diagrams have limitation on presenting more than three variable systems on a two-dimensional plane, hypergraphs in the field of Higher-order Networks may be a better tool for visualizing SID frameworks.

Another field where SID may interact is Causal Science, since it is a field for studying the intrinsic relationships between multiple variables. One of the goals of causal science is to search for invariance in the system. We hope that the revealed properties of the system are independent of the distribution of the data. However, the results obtained from SID can vary with changes in the data distribution. Therefore, adopting the methods of causal science in SID to reveal system invariance is one direction worth to explore. In addition, conditional independence plays an important role in causal discovery and causal inference in multivariate systems \cite{pearl2009causality}, while in the quantitative calculation of SID, conditional independence also plays a similar role in eliminating the uncertainty of higher-order relations, refer to the calculation method in the first way. Therefore, studying the properties of conditional independence within the framework of SID may provide a bridge between causal science and SID. The benefits of this association are mutual: from the perspective of Pearl Causal Hierarchy theory \cite{pearl2018book}, SID is a research technique that utilizes observational data, which is at the lowest rung of the causal ladder. Investigating whether lifting the approach to higher rungs of causal ladder can yield deeper insights into the system is an area worth exploring, for instance, by incorporating causal graphs (DAGs) into SID methods, etc. Furthermore, the symmetry property observed in SID highlights that synergistic and redundant information are characteristics that emerge at the system level. This insight guides us to leverage these properties when constructing more powerful models in causal representation learning and interpreting latent variables in cross-level causal analysis.

Apart from the above fields, SID may also has potential applications. Since information atoms provide a more refined division of information entropy, when the physical meaning of information atoms within the SID framework is revealed, specific information atoms may also become indicators for some optimization or learning problems; The symmetry property of synergistic information in SID may provide inspiration for the information disclosure, an important application of PID in information protection field. In summary, SID, as a progress in the underlying measurement, may play a role in many application scenarios, which is also the focus of our next stage of work.

\subsection{Limitations and Future Works}
In addition to the above-mentioned promising progress and expectations, there are still several limitations worthy of attention. The first limitation is the absence of a fully compatible quantitative method for the proposed framework, which restricts the practical application of SID in addressing real-world problems. As we continue to develop and refine the SID framework, it is a priority to develop robust computation methods for calculating SID components and consider how higher-order information measures can be integrated into existing analytical approaches. Furthermore, the existing proofs of framework properties and computational methods have only been established for three-variable systems. Although extending current work to general multivariate systems is not a formidable challenge, it contains many aspects of work, such as how to present the decomposition results of multivariate systems on a two-dimensional plane; how to optimize the calculation algorithm to avoid the exponential calculation cost as the number of variables increases, which will be considered in the next stage of research. For the above-mentioned and any possible problems, we cordially invite other scholars who share an interest in this field to collaborate on addressing the existing challenges of SID and contribute to the model's refinement.

\section{Conclusion}
\label{sec:Conclusion}

In this study, we introduced the System Information Decomposition (SID) framework, which offers novel insights for decomposing complex systems and analyzing higher-order relationships while addressing the limitations of existing information decomposition methods.By proving the symmetries of information atoms and connecting them to higher-order relationships, we show that the SID framework can provide insights and advance beyond existing measures in understanding the internal interactions and dynamics of complex systems.Furthermore, we explored the far-reaching implications that SID's unveiling of higher-order measures could have on the philosophical aspects of systems research, higher-order networks, and causal science. Despite the fact that current research on SID still faces challenges in terms of quantitative calculations and multivariate analysis, we believe that continued collaboration and exploration by the scientific community will help overcome these obstacles.In conclusion, the SID framework signifies a promising new direction for investigating complex systems and information decomposition. We anticipate that the SID analysis framework will serve as a valuable tool across an expanding array of fields in the future.

\section*{Acknowledgments}
We sincerely thank all the non-authors who played a crucial role in its successful completion. Our heartfelt appreciation goes to the Swarma Club, an open academic community for Complex Systems, where the Causal Emergence reading club provided the foundation for the ideas presented in this paper. We are also very grateful to Professor Duguid at UC Berkeley, whose course steadied an author's orientation towards understanding systems from an information perspective, serving as the genesis of this paper. We are also very grateful to the reviewers for their constructive comments, which have improved the theoretical rigor and comprehensiveness of the paper.

\newpage
\bibliographystyle{unsrt}
\bibliography{references}

\newpage
\appendix
\section{Appendix}
\label{Appendix}

\subsection{Case Table}
\label{Table:Case Table}

\setlength{\arrayrulewidth}{0.7pt}
\renewcommand{\arraystretch}{0.71}
\begin{longtable}{ | c c c c | c c c c | c c c c | c c c c | c c c c | c c c c | }

\hline
\multicolumn{4}{|c|}{$X_1$} & \multicolumn{4}{c|}{$X_2$} & \multicolumn{4}{c|}{$X_3$} & \multicolumn{4}{c|}{$X_4$} & \multicolumn{4}{c|}{$X_5$} & \multicolumn{4}{c|}{$X_6$} \\
\hline
$a$ & $b$ & $c$ & $d$ & $a$ & $b$ & $e$ & $f$ & $c$ & $d$ & $e$ & $f$ & $a$ & $c$ & $e$ & $h$ & $a$ & $b$ & $g$ & $h$ & $a$ & $b$ & $i$ & $j$ \\
\hline
\rowcolor{yellow}
0 & 0 & 0 & 0 & 0 & 0 & 0 & 0 & 0 & 0 & 0 & 0 & 0 & 0 & 0 & 0 & 0 & 0 & 0 & 0 & 0 & 0 & 0 & 0 \\
\rowcolor{yellow}
0 & 0 & 0 & 0 & 0 & 0 & 0 & 1 & 0 & 0 & 0 & 1 & 0 & 0 & 0 & 1 & 0 & 0 & 0 & 1 & 0 & 0 & 1 & 0 \\
\rowcolor{yellow}
0 & 0 & 0 & 0 & 0 & 0 & 1 & 0 & 0 & 0 & 1 & 0 & 0 & 0 & 1 & 0 & 0 & 0 & 1 & 0 & 0 & 0 & 0 & 1 \\
\rowcolor{yellow}
0 & 0 & 0 & 0 & 0 & 0 & 1 & 1 & 0 & 0 & 1 & 1 & 0 & 0 & 1 & 1 & 0 & 0 & 1 & 1 & 0 & 0 & 1 & 1 \\
0 & 0 & 0 & 1 & \cellcolor{pink}\textbf{0} & \cellcolor{pink}\textbf{0} & \cellcolor{pink}\textbf{0} & \cellcolor{pink}\textbf{0} & 0 & 1 & 0 & 0 & \cellcolor{pink}\textbf{0} & \cellcolor{pink}\textbf{0} & \cellcolor{pink}\textbf{0} & \cellcolor{pink}\textbf{1} & \cellcolor{pink}0 & \cellcolor{pink}0 & \cellcolor{pink}0 & \cellcolor{pink}1 & \cellcolor{pink}0 & \cellcolor{pink}0 & \cellcolor{pink}0 & \cellcolor{pink}1 \\
0 & 0 & 0 & 1 & \cellcolor{pink}\textbf{0} & \cellcolor{pink}\textbf{0} & \cellcolor{pink}\textbf{0} & \cellcolor{pink}\textbf{1} & 0 & 1 & 0 & 1 & \cellcolor{pink}\textbf{0} & \cellcolor{pink}\textbf{0} & \cellcolor{pink}\textbf{0} & \cellcolor{pink}\textbf{0} & \cellcolor{pink}0 & \cellcolor{pink}0 & \cellcolor{pink}0 & \cellcolor{pink}0 & \cellcolor{pink}0 & \cellcolor{pink}0 & \cellcolor{pink}1 & \cellcolor{pink}1 \\
0 & 0 & 0 & 1 & \cellcolor{pink}\textbf{0} & \cellcolor{pink}\textbf{0} & \cellcolor{pink}\textbf{1} & \cellcolor{pink}\textbf{0} & 0 & 1 & 1 & 0 & \cellcolor{pink}\textbf{0} & \cellcolor{pink}\textbf{0} & \cellcolor{pink}\textbf{1} & \cellcolor{pink}\textbf{1} & \cellcolor{pink}0 & \cellcolor{pink}0 & \cellcolor{pink}1 & \cellcolor{pink}1 & \cellcolor{pink}0 & \cellcolor{pink}0 & \cellcolor{pink}0 & \cellcolor{pink}0 \\
0 & 0 & 0 & 1 & \cellcolor{pink}\textbf{0} & \cellcolor{pink}\textbf{0} & \cellcolor{pink}\textbf{1} & \cellcolor{pink}\textbf{1} & 0 & 1 & 1 & 1 & \cellcolor{pink}\textbf{0} & \cellcolor{pink}\textbf{0} & \cellcolor{pink}\textbf{1} & \cellcolor{pink}\textbf{0} & \cellcolor{pink}0 & \cellcolor{pink}0 & \cellcolor{pink}1 & \cellcolor{pink}0 & \cellcolor{pink}0 & \cellcolor{pink}0 & \cellcolor{pink}1 & \cellcolor{pink}0 \\
0 & 0 & 1 & 0 & \cellcolor{pink}\textit{0} & \cellcolor{pink}\textit{0} & \cellcolor{pink}\textit{0} & \cellcolor{pink}\textit{0} & 1 & 0 & 0 & 0 & 0 & 1 & 0 & 0 & \cellcolor{pink}0 & \cellcolor{pink}0 & \cellcolor{pink}1 & \cellcolor{pink}0 & \cellcolor{pink}0 & \cellcolor{pink}0 & \cellcolor{pink}1 & \cellcolor{pink}0 \\
0 & 0 & 1 & 0 & \cellcolor{pink}\textit{0} & \cellcolor{pink}\textit{0} & \cellcolor{pink}\textit{0} & \cellcolor{pink}\textit{1} & 1 & 0 & 0 & 1 & 0 & 1 & 0 & 1 & \cellcolor{pink}0 & \cellcolor{pink}0 & \cellcolor{pink}1 & \cellcolor{pink}1 & \cellcolor{pink}0 & \cellcolor{pink}0 & \cellcolor{pink}0 & \cellcolor{pink}0 \\
0 & 0 & 1 & 0 & \cellcolor{pink}\textit{0} & \cellcolor{pink}\textit{0} & \cellcolor{pink}\textit{1} & \cellcolor{pink}\textit{0} & 1 & 0 & 1 & 0 & 0 & 1 & 1 & 0 & \cellcolor{pink}0 & \cellcolor{pink}0 & \cellcolor{pink}0 & \cellcolor{pink}0 & \cellcolor{pink}0 & \cellcolor{pink}0 & \cellcolor{pink}1 & \cellcolor{pink}1 \\
0 & 0 & 1 & 0 & \cellcolor{pink}\textit{0} & \cellcolor{pink}\textit{0} & \cellcolor{pink}\textit{1} & \cellcolor{pink}\textit{1} & 1 & 0 & 1 & 1 & 0 & 1 & 1 & 1 & \cellcolor{pink}0 & \cellcolor{pink}0 & \cellcolor{pink}0 & \cellcolor{pink}1 & \cellcolor{pink}0 & \cellcolor{pink}0 & \cellcolor{pink}0 & \cellcolor{pink}1 \\
0 & 0 & 1 & 1 & \cellcolor{pink}\textit{0} & \cellcolor{pink}\textit{0} & \cellcolor{pink}\textit{0} & \cellcolor{pink}\textit{0} & 1 & 1 & 0 & 0 & 0 & 1 & 0 & 1 & \cellcolor{pink}0 & \cellcolor{pink}0 & \cellcolor{pink}1 & \cellcolor{pink}1 & \cellcolor{pink}0 & \cellcolor{pink}0 & \cellcolor{pink}1 & \cellcolor{pink}1 \\
0 & 0 & 1 & 1 & \cellcolor{pink}\textit{0} & \cellcolor{pink}\textit{0} & \cellcolor{pink}\textit{0} & \cellcolor{pink}\textit{1} & 1 & 1 & 0 & 1 & 0 & 1 & 0 & 0 & \cellcolor{pink}0 & \cellcolor{pink}0 & \cellcolor{pink}1 & \cellcolor{pink}0 & \cellcolor{pink}0 & \cellcolor{pink}0 & \cellcolor{pink}0 & \cellcolor{pink}1 \\
0 & 0 & 1 & 1 & \cellcolor{pink}\textit{0} & \cellcolor{pink}\textit{0} & \cellcolor{pink}\textit{1} & \cellcolor{pink}\textit{0} & 1 & 1 & 1 & 0 & 0 & 1 & 1 & 1 & \cellcolor{pink}0 & \cellcolor{pink}0 & \cellcolor{pink}0 & \cellcolor{pink}1 & \cellcolor{pink}0 & \cellcolor{pink}0 & \cellcolor{pink}1 & \cellcolor{pink}0 \\
0 & 0 & 1 & 1 & \cellcolor{pink}\textit{0} & \cellcolor{pink}\textit{0} & \cellcolor{pink}\textit{1} & \cellcolor{pink}\textit{1} & 1 & 1 & 1 & 1 & 0 & 1 & 1 & 0 & \cellcolor{pink}0 & \cellcolor{pink}0 & \cellcolor{pink}0 & \cellcolor{pink}0 & \cellcolor{pink}0 & \cellcolor{pink}0 & \cellcolor{pink}0 & \cellcolor{pink}0 \\
0 & 1 & 0 & 0 & 0 & 1 & 0 & 0 & \cellcolor{pink}0 & \cellcolor{pink}0 & \cellcolor{pink}0 & \cellcolor{pink}0 & \cellcolor{pink}\textit{0} & \cellcolor{pink}\textit{0} & \cellcolor{pink}\textit{0} & \cellcolor{pink}\textit{0} & 0 & 1 & 0 & 0 & 0 & 1 & 0 & 0 \\
0 & 1 & 0 & 0 & 0 & 1 & 0 & 1 & \cellcolor{pink}0 & \cellcolor{pink}0 & \cellcolor{pink}0 & \cellcolor{pink}1 & \cellcolor{pink}\textit{0} & \cellcolor{pink}\textit{0} & \cellcolor{pink}\textit{0} & \cellcolor{pink}\textit{1} & 0 & 1 & 0 & 1 & 0 & 1 & 1 & 0 \\
0 & 1 & 0 & 0 & 0 & 1 & 1 & 0 & \cellcolor{pink}0 & \cellcolor{pink}0 & \cellcolor{pink}1 & \cellcolor{pink}0 & \cellcolor{pink}\textit{0} & \cellcolor{pink}\textit{0} & \cellcolor{pink}\textit{1} & \cellcolor{pink}\textit{0} & 0 & 1 & 1 & 0 & 0 & 1 & 0 & 1 \\
0 & 1 & 0 & 0 & 0 & 1 & 1 & 1 & \cellcolor{pink}0 & \cellcolor{pink}0 & \cellcolor{pink}1 & \cellcolor{pink}1 & \cellcolor{pink}\textit{0} & \cellcolor{pink}\textit{0} & \cellcolor{pink}\textit{1} & \cellcolor{pink}\textit{1} & 0 & 1 & 1 & 1 & 0 & 1 & 1 & 1 \\
0 & 1 & 0 & 1 & 0 & 1 & 0 & 0 & 0 & 1 & 0 & 0 & \cellcolor{pink}\textit{0} & \cellcolor{pink}\textit{0} & \cellcolor{pink}\textit{0} & \cellcolor{pink}\textit{1} & 0 & 1 & 0 & 1 & 0 & 1 & 0 & 1 \\
0 & 1 & 0 & 1 & 0 & 1 & 0 & 1 & 0 & 1 & 0 & 1 & \cellcolor{pink}\textit{0} & \cellcolor{pink}\textit{0} & \cellcolor{pink}\textit{0} & \cellcolor{pink}\textit{0} & 0 & 1 & 0 & 0 & 0 & 1 & 1 & 1 \\
0 & 1 & 0 & 1 & 0 & 1 & 1 & 0 & 0 & 1 & 1 & 0 & \cellcolor{pink}\textit{0} & \cellcolor{pink}\textit{0} & \cellcolor{pink}\textit{1} & \cellcolor{pink}\textit{1} & 0 & 1 & 1 & 1 & 0 & 1 & 0 & 0 \\
0 & 1 & 0 & 1 & 0 & 1 & 1 & 1 & 0 & 1 & 1 & 1 & \cellcolor{pink}\textit{0} & \cellcolor{pink}\textit{0} & \cellcolor{pink}\textit{1} & \cellcolor{pink}\textit{0} & 0 & 1 & 1 & 0 & 0 & 1 & 1 & 0 \\
0 & 1 & 1 & 0 & 0 & 1 & 0 & 0 & 1 & 0 & 0 & 0 & 0 & 1 & 0 & 0 & 0 & 1 & 1 & 0 & 0 & 1 & 1 & 0 \\
0 & 1 & 1 & 0 & 0 & 1 & 0 & 1 & 1 & 0 & 0 & 1 & 0 & 1 & 0 & 1 & 0 & 1 & 1 & 1 & 0 & 1 & 0 & 0 \\
0 & 1 & 1 & 0 & 0 & 1 & 1 & 0 & 1 & 0 & 1 & 0 & 0 & 1 & 1 & 0 & 0 & 1 & 0 & 0 & 0 & 1 & 1 & 1 \\
0 & 1 & 1 & 0 & 0 & 1 & 1 & 1 & 1 & 0 & 1 & 1 & 0 & 1 & 1 & 1 & 0 & 1 & 0 & 1 & 0 & 1 & 0 & 1 \\
0 & 1 & 1 & 1 & 0 & 1 & 0 & 0 & 1 & 1 & 0 & 0 & 0 & 1 & 0 & 1 & 0 & 1 & 1 & 1 & 0 & 1 & 1 & 1 \\
0 & 1 & 1 & 1 & 0 & 1 & 0 & 1 & 1 & 1 & 0 & 1 & 0 & 1 & 0 & 0 & 0 & 1 & 1 & 0 & 0 & 1 & 0 & 1 \\
0 & 1 & 1 & 1 & 0 & 1 & 1 & 0 & 1 & 1 & 1 & 0 & 0 & 1 & 1 & 1 & 0 & 1 & 0 & 1 & 0 & 1 & 1 & 0 \\
0 & 1 & 1 & 1 & 0 & 1 & 1 & 1 & 1 & 1 & 1 & 1 & 0 & 1 & 1 & 0 & 0 & 1 & 0 & 0 & 0 & 1 & 0 & 0 \\
1 & 0 & 0 & 0 & 1 & 0 & 0 & 0 & \cellcolor{pink}0 & \cellcolor{pink}0 & \cellcolor{pink}0 & \cellcolor{pink}0 & 1 & 0 & 0 & 0 & 1 & 0 & 0 & 0 & 1 & 0 & 0 & 0 \\
1 & 0 & 0 & 0 & 1 & 0 & 0 & 1 & \cellcolor{pink}0 & \cellcolor{pink}0 & \cellcolor{pink}0 & \cellcolor{pink}1 & 1 & 0 & 0 & 1 & 1 & 0 & 0 & 1 & 1 & 0 & 1 & 0 \\
1 & 0 & 0 & 0 & 1 & 0 & 1 & 0 & \cellcolor{pink}0 & \cellcolor{pink}0 & \cellcolor{pink}1 & \cellcolor{pink}0 & 1 & 0 & 1 & 0 & 1 & 0 & 1 & 0 & 1 & 0 & 0 & 1 \\
1 & 0 & 0 & 0 & 1 & 0 & 1 & 1 & \cellcolor{pink}0 & \cellcolor{pink}0 & \cellcolor{pink}1 & \cellcolor{pink}1 & 1 & 0 & 1 & 1 & 1 & 0 & 1 & 1 & 1 & 0 & 1 & 1 \\
1 & 0 & 0 & 1 & 1 & 0 & 0 & 0 & 0 & 1 & 0 & 0 & 1 & 0 & 0 & 1 & 1 & 0 & 0 & 1 & 1 & 0 & 0 & 1 \\
1 & 0 & 0 & 1 & 1 & 0 & 0 & 1 & 0 & 1 & 0 & 1 & 1 & 0 & 0 & 0 & 1 & 0 & 0 & 0 & 1 & 0 & 1 & 1 \\
1 & 0 & 0 & 1 & 1 & 0 & 1 & 0 & 0 & 1 & 1 & 0 & 1 & 0 & 1 & 1 & 1 & 0 & 1 & 1 & 1 & 0 & 0 & 0 \\
1 & 0 & 0 & 1 & 1 & 0 & 1 & 1 & 0 & 1 & 1 & 1 & 1 & 0 & 1 & 0 & 1 & 0 & 1 & 0 & 1 & 0 & 1 & 0 \\
1 & 0 & 1 & 0 & 1 & 0 & 0 & 0 & 1 & 0 & 0 & 0 & 1 & 1 & 0 & 0 & 1 & 0 & 1 & 0 & 1 & 0 & 1 & 0 \\
1 & 0 & 1 & 0 & 1 & 0 & 0 & 1 & 1 & 0 & 0 & 1 & 1 & 1 & 0 & 1 & 1 & 0 & 1 & 1 & 1 & 0 & 0 & 0 \\
1 & 0 & 1 & 0 & 1 & 0 & 1 & 0 & 1 & 0 & 1 & 0 & 1 & 1 & 1 & 0 & 1 & 0 & 0 & 0 & 1 & 0 & 1 & 1 \\
1 & 0 & 1 & 0 & 1 & 0 & 1 & 1 & 1 & 0 & 1 & 1 & 1 & 1 & 1 & 1 & 1 & 0 & 0 & 1 & 1 & 0 & 0 & 1 \\
1 & 0 & 1 & 1 & 1 & 0 & 0 & 0 & 1 & 1 & 0 & 0 & 1 & 1 & 0 & 1 & 1 & 0 & 1 & 1 & 1 & 0 & 1 & 1 \\
1 & 0 & 1 & 1 & 1 & 0 & 0 & 1 & 1 & 1 & 0 & 1 & 1 & 1 & 0 & 0 & 1 & 0 & 1 & 0 & 1 & 0 & 0 & 1 \\
1 & 0 & 1 & 1 & 1 & 0 & 1 & 0 & 1 & 1 & 1 & 0 & 1 & 1 & 1 & 1 & 1 & 0 & 0 & 1 & 1 & 0 & 1 & 0 \\
1 & 0 & 1 & 1 & 1 & 0 & 1 & 1 & 1 & 1 & 1 & 1 & 1 & 1 & 1 & 0 & 1 & 0 & 0 & 0 & 1 & 0 & 0 & 0 \\
1 & 1 & 0 & 0 & 1 & 1 & 0 & 0 & \cellcolor{pink}0 & \cellcolor{pink}0 & \cellcolor{pink}0 & \cellcolor{pink}0 & 1 & 0 & 0 & 0 & 1 & 1 & 0 & 0 & 1 & 1 & 0 & 0 \\
1 & 1 & 0 & 0 & 1 & 1 & 0 & 1 & \cellcolor{pink}0 & \cellcolor{pink}0 & \cellcolor{pink}0 & \cellcolor{pink}1 & 1 & 0 & 0 & 1 & 1 & 1 & 0 & 1 & 1 & 1 & 1 & 0 \\
1 & 1 & 0 & 0 & 1 & 1 & 1 & 0 & \cellcolor{pink}0 & \cellcolor{pink}0 & \cellcolor{pink}1 & \cellcolor{pink}0 & 1 & 0 & 1 & 0 & 1 & 1 & 1 & 0 & 1 & 1 & 0 & 1 \\
1 & 1 & 0 & 0 & 1 & 1 & 1 & 1 & \cellcolor{pink}0 & \cellcolor{pink}0 & \cellcolor{pink}1 & \cellcolor{pink}1 & 1 & 0 & 1 & 1 & 1 & 1 & 1 & 1 & 1 & 1 & 1 & 1 \\
1 & 1 & 0 & 1 & 1 & 1 & 0 & 0 & 0 & 1 & 0 & 0 & 1 & 0 & 0 & 1 & 1 & 1 & 0 & 1 & 1 & 1 & 0 & 1 \\
1 & 1 & 0 & 1 & 1 & 1 & 0 & 1 & 0 & 1 & 0 & 1 & 1 & 0 & 0 & 0 & 1 & 1 & 0 & 0 & 1 & 1 & 1 & 1 \\
1 & 1 & 0 & 1 & 1 & 1 & 1 & 0 & 0 & 1 & 1 & 0 & 1 & 0 & 1 & 1 & 1 & 1 & 1 & 1 & 1 & 1 & 0 & 0 \\
1 & 1 & 0 & 1 & 1 & 1 & 1 & 1 & 0 & 1 & 1 & 1 & 1 & 0 & 1 & 0 & 1 & 1 & 1 & 0 & 1 & 1 & 1 & 0 \\
1 & 1 & 1 & 0 & 1 & 1 & 0 & 0 & 1 & 0 & 0 & 0 & 1 & 1 & 0 & 0 & 1 & 1 & 1 & 0 & 1 & 1 & 1 & 0 \\
1 & 1 & 1 & 0 & 1 & 1 & 0 & 1 & 1 & 0 & 0 & 1 & 1 & 1 & 0 & 1 & 1 & 1 & 1 & 1 & 1 & 1 & 0 & 0 \\
1 & 1 & 1 & 0 & 1 & 1 & 1 & 0 & 1 & 0 & 1 & 0 & 1 & 1 & 1 & 0 & 1 & 1 & 0 & 0 & 1 & 1 & 1 & 1 \\
1 & 1 & 1 & 0 & 1 & 1 & 1 & 1 & 1 & 0 & 1 & 1 & 1 & 1 & 1 & 1 & 1 & 1 & 0 & 1 & 1 & 1 & 0 & 1 \\
1 & 1 & 1 & 1 & 1 & 1 & 0 & 0 & 1 & 1 & 0 & 0 & 1 & 1 & 0 & 1 & 1 & 1 & 1 & 1 & 1 & 1 & 1 & 1 \\
1 & 1 & 1 & 1 & 1 & 1 & 0 & 1 & 1 & 1 & 0 & 1 & 1 & 1 & 0 & 0 & 1 & 1 & 1 & 0 & 1 & 1 & 0 & 1 \\
1 & 1 & 1 & 1 & 1 & 1 & 1 & 0 & 1 & 1 & 1 & 0 & 1 & 1 & 1 & 1 & 1 & 1 & 0 & 1 & 1 & 1 & 1 & 0 \\
1 & 1 & 1 & 1 & 1 & 1 & 1 & 1 & 1 & 1 & 1 & 1 & 1 & 1 & 1 & 0 & 1 & 1 & 0 & 0 & 1 & 1 & 0 & 0 \\
\hline
\end{longtable}

\subsection{Proof of Propositions for Neural Information Squeezer Network}

Here we provide mathematical proves for the two propositions of the neural network framework to calculate mutual information and redundancy.

First, we rephrase the proposition 1 and then we give the proof here.\\

\textbf{Proposition 1}: For any random variables $X$ and $Y$, we can use the framework of Figure \ref{fig:nis}(a) to predict $Y$ by squeezing the information channel of $\hat{Y}_X$ as the minimum dimension but satisfying $\hat{Y}\approx Y$ and $U\sim\mathcal{N}(0,I)$. And we suppose the conditional entropy $H(X|Y)>0$ holds, then:
\begin{equation}
    H(\hat{Y}_X)\approx I(X;Y)
\end{equation}

\begin{proof}
The whole structure of the alternative NIS network (Figure \ref{fig:nis}(a)) can be regarded as the similar structure as in Ref \cite{zhang2023nis}, but the dynamics learner is absent. However, we can understand the dynamic is a fixed identical mapping. In this way, all the conclusions proved in \cite{zhang2023nis} can be applied here. Thus, we have:
\begin{equation}
\label{eqn.first_nis}
    I(X;Y)\approx I(\hat{Y}_X;\hat{Y}_X)=H(\hat{Y}_X)
\end{equation}
if all the neural networks are well trained.
The first equation holds because of Theorem 2 (information bottle-neck) and Theorem 3(mutual information of the model will be closed to the data for a well trained framework) in \cite{zhang2023nis}, the second holds when $q$ is minimized such that the information channel of $\hat{Y}_X$ is squeezed as possible as we can and because of the property of mutual information. 

Further, because $U$ is an independent Gaussian noise, therefore:

\begin{equation}
\label{eqn.noise_nis}
    H(U)=H(\psi(X))-H(\hat{Y_X})\approx H(X)-I(X;Y)=H(X|Y)
\end{equation}

The approximated equation holds because $\psi$ is a bijector which can keep the entropy unchanged, and Equation \ref{eqn.first_nis} holds. Therefore, we can prove proposition 1.
\end{proof}

To calculate the redundancy for a system with three variables we can further feed the variable of $\hat{Y}_X$ into another NIS network to predict $Z$, and narrow down the information channel of the intermediate variable $\hat{Z}_{\hat{Y}_X}$ to get the minimum dimension $q^{*'}$ for $\hat{Z}_{\hat{Y}_X}$, then its Shannon entropy can approach the redundancy, and the redundancy satisfies the property of permutation symmetry for all the variables. We can prove the following proposition:

\textbf{Proposition 2}: For a system with three random variables $X,Y,Z$, suppose the conditional information $H(X|Y)>0,H(X|Z)>0$, then the redundancy calculated by Equation \ref{eqn.nis_red} is symmetric, which means:
\begin{equation}
    Red(X,Y,Z)\approx Red(X,Z,Y)
    \end{equation}

\begin{proof}
If we accept the definition of Equation \ref{eqn.nis_red}, then:
\begin{equation}
\label{eqn.red1}
    Red(X,Y,Z)\approx H(\hat{Z}_{\hat{Y}_X})=H(\hat{Y}_X)-H(U_{\hat{Y}_X})=H(X)-H(X|Y)-H(\hat{Y}_X|Z),
\end{equation}
where $U_{\hat{Y}_X}$ is the discarded Gaussian noise to predict $\hat{Y}_Z$.

In another way, we can use $X$ to predict $Z$, and the intermediate variable $\hat{Z}_X$ can be used to predict $Y$, and the intermediate variable $\hat{Y}_{\hat{Z}_X}$ can be used to approximate the redundancy which is denoted as $Red(X,Z,Y)$. Therefore,
\begin{equation}
\label{eqn.red2}
    Red(X,Z,Y)\approx H(X)-H(X|Z)-H(\hat{Z}_X|Y).
\end{equation}

Because the discarded noise variable $U_{\hat{Y}_X}$ in the process of predicting $Y$ by $X$ is independent on all the variables, therefore:
\begin{equation}
    H(U_{\hat{Y}_X})=H(U_{\hat{Y}_X}|Z)=H(U_{\hat{Y}_X}|Y,Z)=H(X|Y,Z),
\end{equation}

Similarly, the discarded noise variable $U_{\hat{Z}_{\hat{Y}_X}}$ in the process of predicting $Z$ by $\hat{Y}_X$ is also independent on all the other variables, and $\psi(X)$ is the combination of $U_{\hat{Y}_X}$ and $\hat{Y}_X$, thus:

\begin{equation}
    H(X|Y,Z)=H(U_{\hat{Y}_X}|Z)=H(X|Z)-H(\hat{Y}_X|Z).
\end{equation}

In the same way, we can obtain:
\begin{equation}
    H(X|Z,Y)=H(U_{\hat{Z}_X}|Y)=H(X|Y)-H(\hat{Z}_X|Y).
\end{equation}
Because $H(X|Y,Z)=H(X|Z)-H(\hat{Y}_X|Z)=H(X|Y,Z)=H(X|Y)-H(\hat{Z}_X|Y)$, therefore:

\begin{equation}
    H(X|Z)+H(\hat{Z}_X|Y)=H(X|Y)+H(\hat{Y}_X|Z)
\end{equation}
and the Equation \ref{eqn.red1} and \ref{eqn.red2} lead to:
\begin{equation}
    Red(X,Y,Z)=Red(X,Z,Y).
\end{equation}

This equation is general for all the permutations of $X,Y$ and $Z$, thus, the redundancy defined in the neural network NIS satisfies permutation symmetry.

\end{proof}

\subsection{Proof of Theorem 1 with different partial order}
\label{Proof of Theorem 1 with different partial order}

To proof the Theorem \ref{Theorem:Symmetry of Redundant Information}, we can also use the Definition \ref{definition:Set Intersection of Information} with another partial order, such that $Q \sqsubset X $ if and only if $H(Q|X) = 0$.

\begin{proof}
Suppose we have a multivariate system containing a target variable $Y$ and source variables $X_{1},\cdots , X_{n}$. For the convenience of expression, we use $\mathcal{X}$ to represent all the source variables $X_{1},\cdots , X_{n}$. The proof is to show that $Red(Y: \mathcal{X},Y) = Red(Y;\mathcal{X})$ and $ Red(U: \mathcal{X}, Y) = Red(Y: \mathcal{X}, Y) $, where $U$ is the union variable of $Y$ and $\mathcal{X}$, such that $U = (\mathcal{X}, Y)$. Then, we can demonstrate that redundant information is equal regardless of which variable is chosen as the target variable.

\underline{Step One, to prove $Red(Y: \mathcal{X},Y) = Red(Y;\mathcal{X}):$} 

By Definition \ref{definition:Set Intersection of Information}, 
\begin{align}
Red(Y: \mathcal{X},Y) =\sup_{Q_{j}} \{I(Q_{j}:Y):Q_{j}\sqsubset Y, Q_{j}\sqsubset X_{i},\forall i\in \{1\cdots n\} \} 
\end{align} 
According to the Monotonicity property of redundant information (Axiom \ref{Axiom:Monotonicity}) that adding new source variables will only impose stricter restrictions on top of existing ones, and the Symmetry property of source variables (Axiom \ref{Axiom:Symmetry of source variables}) that the order in which restrictions are imposed will not affect the results, we can make this optimization problem into two steps, such that:

$\sup_{Q_{j}} \{I(Q_{j}:Y):Q_{j}\sqsubset Y, Q_{j}\sqsubset X_{i},\forall i\in \{1\cdots n\} \}$ 

$= \sup_{Q_{j},Q_{k}} \{I(Q_{j}:Y):Q_{j}\sqsubset Y, Q_{j}\sqsubset Q_{k}, Q_{k}\sqsubset X_{i},\forall i\in \{1\cdots n\} \} $

$= \sup_{Q_{k}} \{ \sup_{Q_{j}} \{I(Q_{j}:Y):Q_{j}\sqsubset Y, Q_{j}\sqsubset Q_{k} \}: Q_{k} \sqsubset X_{i},\forall i\in \{1\cdots n\} \} $

$= \sup_{Q_{k}} \{ \sup_{Q_{j}} \{H(Q_{j}):Q_{j}\sqsubset Y, Q_{j}\sqsubset Q_{k} \}: Q_{k} \sqsubset X_{i},\forall i\in \{1\cdots n\} \}$, since $Q_{j}\sqsubset Y$

$= \sup_{Q_{k}} \{I(Q_{k}:Y):Q_{k}\sqsubset X_{i},\forall i\in \{1\cdots n\}\}$, since $ \sup_{Q_{j}} \{H(Q_{j}):Q_{j}\sqsubset Y, Q_{j}\sqsubset Q_{k} \} = I(Q_{k}:Y)$. 


Therefore, $Red(Y: \mathcal{X},Y) =
Red(Y;\mathcal{X})$

\underline{Step Two, to prove $ Red(U: \mathcal{X}, Y) = Red(Y: \mathcal{X}, Y) $:}
Building upon the conclusion that $Red(Y: \mathcal{X}, Y) = Red(Y: \mathcal{X})$, we can replace the target variable with the union variable $U = (\mathcal{X}, Y)$, which combines the target variable $Y$ and the source variables $\mathcal{X}$. (The entropy of the union variable $U$ can be expressed as $H(U) = H(\mathcal{X}, Y)$.)

Firstly, let's employ the contradiction method by assuming that $Red(U: \mathcal{X}, Y) < Red(Y: \mathcal{X}, Y)$.
That means that: 
\begin{align}
\sup_{Q_{j}} \{I(Q_{j}:U):Q_{j}\sqsubset Y, Q_{j}\sqsubset X_{i},\forall i\in \{1\cdots n\} \} < \sup_{Q_{k}}\{I(Q_{k}:Y): Q_{k}\sqsubset Y, Q_{k}\sqsubset X_{i},\forall i\in \{1\cdots n\} \}
\end{align}
Let $Q_{j}^{*}$ and $Q_{k}^{*}$ satisfies or infinitely approaches the above conditions:
\begin{align}
I(Q_{j}^{*}:U) &= Red(U: \mathcal{X}, Y) - \varepsilon, \forall \varepsilon > 0 \nonumber \\ &= \sup_{Q_{j}} \{I(Q_{j}:U):Q_{j}\sqsubset Y, Q_{j}\sqsubset X_{i},\forall i\in \{1\cdots n\}\} - \varepsilon, \forall \varepsilon > 0, \nonumber 
\end{align}
\begin{align}
I(Q_{k}^{*}:Y) &= Red(Y: \mathcal{X}, Y) - \varepsilon, \forall \varepsilon > 0 \nonumber \\ &= \sup_{Q_{k}} \{I(Q_{k}:Y):Q_{k}\sqsubset Y, Q_{k}\sqsubset X_{i},\forall i\in \{1\cdots n\}\} - \varepsilon, \forall \varepsilon > 0, \nonumber 
\end{align}
Since $Y \sqsubset U$ from $U=(\mathcal{X},Y)$, we have:
\begin{align}
I(Q_{k}^{*},Y) \le I(Q_{k}^{*},U)
\end{align}
Given that $Q_{k}^{*}\sqsubset Y$ and $Q_{k}^{*}\sqsubset X_{i}$ (same with the restrictions of $Q_{j}^{*}$), the mutual information $I(Q_{j}^{*},U)$ should greater or equal to $I(Q_{k}^{*},Y)$, which lead to a contradiction. 
Consequently, we can conclude that $Red(U: \mathcal{X}, Y) \ge Red(Y: \mathcal{X}, Y)$.

Secondly, let's also use the contradiction method by assuming that $Red(U: \mathcal{X}, Y) > Red(Y: \mathcal{X}, Y)$. 
In this case:
\begin{align}
\sup_{Q_{j}} \{I(Q_{j}:U):Q_{j}\sqsubset Y, Q_{j}\sqsubset X_{i},\forall i\in \{1\cdots n\} \} > \sup_{Q_{k}}\{I(Q_{k}:Y): Q_{k}\sqsubset Y, Q_{k}\sqsubset X_{i},\forall i\in \{1\cdots n\} \}
\end{align}
Let's focus on the $Q_{j}^{*}$ and $Q_{k}^{*}$ that satisfies or infinitely approaches the above conditions. Since $Q_{j}^{*} \sqsubset Y$ and $Y \sqsubset U$ from $U=(\mathcal{X},Y) (H(Y|U)=0)$, we have:
\begin{align}
I(Q_{j}^{*}:U) = I(Q_{j}^{*}:Y)
\end{align}
That lead to a contradiction, such that $I(Q_{j}^{*}:Y) > I(Q_{k}^{*}:Y) $ with the same restriction on $Q_{j}^{*}$ and $Q_{k}^{*}$.  
Therefore, we obtain $Red(U: \mathcal{X}, Y) \le  Red(Y: \mathcal{X}, Y)$.

Since we have both $Red(U: \mathcal{X}, Y) \ge Red(Y: \mathcal{X}, Y)$ and $Red(U: \mathcal{X}, Y) \le Red(Y: \mathcal{X}, Y)$, $Red(U: \mathcal{X}, Y) = Red(Y: \mathcal{X}, Y)$ is proved.

\underline{In Summary:}
Since we have established that $Red(Y: \mathcal{X}, Y) = Red(Y: \mathcal{X})$, and $ Red(U: \mathcal{X}, Y) = Red(Y: \mathcal{X}, Y)$, we can conclude that for all $X_{i}$ in $\{\mathcal{X}\}$, $Red(X_{i}: Y, \{\mathcal{X}\} \setminus  X_{i}) = Red(Y:\{\mathcal{X}\})$. Therefore, Theorem \ref{Theorem:Symmetry of Redundant Information} is proved, and we can use $Red(X_{1},\cdots, X_{n})$ or $Red_{1\cdots n}$ denote the redundant information within the system $\{X_{1},\cdots, X_{n}\}$.
\end{proof}

\end{document}